\documentclass[a4paper,reqno]{amsart}

\usepackage{amssymb}
\usepackage{latexsym}
\usepackage{amsmath}
\usepackage{amsthm}
\usepackage{bbm}
\usepackage{pgfplots}
\pgfplotsset{compat=1.15}
\usepackage{mathrsfs}
\usetikzlibrary{arrows}
\definecolor{ududff}{rgb}{0.30196078431372547,0.30196078431372547,1}
\definecolor{zzttqq}{rgb}{0.6,0.2,0}
\definecolor{xdxdff}{rgb}{0.49019607843137253,0.49019607843137253,1}
\definecolor{ududff}{rgb}{0.30196078431372547,0.30196078431372547,1}
\usepackage{hyperref}
\usepackage{xcolor}
\usepackage{amsmath}
\allowdisplaybreaks

\def\sa{{\mathfrak a}}
\def\phi{\varphi}

   \def\dN{{\mathbb N}}   
      \def\dR{{\mathbb R}}

\def\cA{{\mathcal A}}   \def\cB{{\mathcal B}}   
\def\cD{{\mathcal D}}   \def\cE{{\mathcal E}}   \def\cF{{\mathcal F}}
\def\cG{{\mathcal G}}

\def\cP{{\mathcal P}}      
      
\def\cV{{\mathcal V}}

\def\dd{{\mathrm d}}
\def\ee{{\mathrm e}}

\DeclareMathOperator\Imag{{\text{\rm Im}}}

\DeclareMathOperator\dom{{\text{\rm dom\,}}} %Domain

\DeclareMathOperator{\diam}{diam}
\DeclareMathOperator{\N}{\mathbb{N}}

\DeclareMathOperator{\C}{\mathbb{C}}

\def\dd{\mathrm{d}}

\DeclareMathOperator{\diag}{diag}

\renewcommand\setminus{\mathbin{\mathpalette\rsetminusaux\relax}}
\newcommand\rsetminusaux[2]{\mspace{-4mu}
  \raisebox{\rsmraise{#1}\depth}{\rotatebox[origin=c]{-20}{$#1\smallsetminus$}}
 \mspace{-4mu}
}
\newcommand\rsmraise[1]{%
  \ifx#1\displaystyle .8\else
    \ifx#1\textstyle .8\else
      \ifx#1\scriptstyle .6\else
        .45%
      \fi
    \fi
  \fi}

\makeatletter
\newcommand*{\rom}[1]{\expandafter\@slowromancap\romannumeral #1@}
\makeatother

\newtheorem{theorem}{Theorem}[section]
\newtheorem*{thm*}{Theorem}
\newtheorem{proposition}[theorem]{Proposition}
\newtheorem{corollary}[theorem]{Corollary}

\theoremstyle{definition}

\theoremstyle{definition}
\newtheorem{definition}[theorem]{Definition}
\newtheorem{example}[theorem]{Example}

\newtheorem{remark}[theorem]{Remark}

\numberwithin{equation}{section}

\title[]{Optimization of quantum graph eigenvalues with preferred orientation vertex conditions}

\author[P.~Exner]{Pavel Exner}
\address{Doppler Institute for Mathematical Physics and Applied Mathematics \\  Czech Technical University \\
B\v rehov\'{a}~7, 11519 Prague \\ Czechia, \\ and Department of Theoretical Physics, NPI \\  Czech Academy of Sciences \\ 25068 \v{R}e\v{z}
near Prague \\ Czechia}
\email{exner@ujf.cas.cz}

\author[J.~Rohleder]{Jonathan Rohleder}
\address{Matematiska institutionen \\ Stockholms universitet \\
106 91 Stockholm \\
Sweden}
\email{jonathan.rohleder@math.su.se}

\begin{document}

\begin{abstract}
We discuss Laplacian spectrum on a finite metric graph with vertex couplings violating the time-reversal invariance. For the class of star graphs we determine, under the condition of a fixed total edge length, the configurations for which the ground state eigenvalue is maximized. Furthermore, for general finite metric graphs we provide upper bounds for all eigenvalues.
\end{abstract}

% \date{\today}
\maketitle

%%%%%%%%%%%%%%%%%%%%%%
\section{Introduction}

Quantum graph is a common shorthand for the Laplacian, or a more general Schr\"o\-dinger operator, on a metric graph. The idea appeared first at the early days of quantum theory as a tool to describe spectral properties of hydrocarbon molecules. However, the real development started only in the eighties being motivated by the progress in solid state physics. It became soon clear that the concept has a rich mathematical content which became a matter of investigation in numerous papers; we refer to the monographs \cite{BK13,KN22,Ku24,M14,P3+05} for an exposition and a rich bibliography. One of the questions often addressed concerned optimization of Laplacian eigenvalues on a finite graph in terms of its metric properties -- see, e.g., \cite{AEHK24p,BL17,BKKM17,BCJ21,BHY23,Fr05,K20,KKMM16,KMN13,KN14,Ni87,Ro17,R22} or \cite[Sec.~13.1]{Ku24}.

One of the peculiar properties of quantum graphs is that in order to get the Laplacian to be a well-defined self-adjoint operator, one has to specify its domain by choosing the conditions matching functions at the graph vertices. The choice of these conditions is by far not unique, at a vertex of degree $N$ there is an $N^2$ parameter family of them; the general form of vertex conditions is usually referred to \cite{KS99,Ha00} but they were known already to Rofe-Beketov \cite{RB69}. The ones most often used are the simplest ones, assuming continuity of the functions at the vertex and a vanishing sum of their derivatives; such conditions are most often labeled as {\em continuity-Kirchhoff}, or simply \emph{Kirchhoff}. With rare exceptions \cite{KKT16,KS19,MR21}, this applies in particular to the optimization problem mentioned above.

Other vertex conditions are also worth attention, both from the physics and mathematics point of view. Motivated by an attempt to use quantum graphs to model the anomalous Hall effect \cite{SK15} a vertex coupling violating the time-reversal invariance was proposed \cite{ET18}. To be more precise, at a vertex $v$ of degree $N$ these conditions are given by the relations
 %----------------%
\begin{align}\label{eq:conditionsIntro}
 \left( F_{j + 1} - F_j \right) + i \left( F_j' + F_{j+1}' \right) = 0, \quad j = 1, \dots, N,
\end{align}
 %----------------%
where for a sufficiently regular function $f$ on the graph, $F_1, \dots, F_N$ denote the values at $v$ of the restrictions of $f$ to each of the edges incident to $v$ and $F_1', \dots, F_N'$ are the corresponding derivatives, taken in the direction from the vertex into the edge; when writing \eqref{eq:conditionsIntro} we use the ``cyclic'' convention to identify $F_{N+1}$ with $F_1$, and equally for the derivatives. Note that at ``loose ends'', i.e.\ vertices of degree one, the conditions \eqref{eq:conditionsIntro} simplify to Neumann conditions. It was found that the coupling \eqref{eq:conditionsIntro} and its generalizations have some unusual spectral and transport properties -- see, e.g., \cite{BEL23,BET22,EL19,ET21}. The aim of this paper is to investigate how these properties will be manifested in the optimization of graph Laplacian eigenvalues.

The first situation which we study in detail is the case of star graphs $\Gamma$, i.e.\ graphs with one central vertex of degree $N \geq 3$ and $N$ edges, each of which connects that central vertex to a vertex of degree one. We prove that the ground state eigenvalue $\lambda_1 (\Gamma)$ (which is always negative) is maximized among all stars of the same total edge length by the equilateral one of degree three or four depending on whether the original edge number was odd or even, respectively. Moreover, these graphs are shown to be the unique maximizers. Secondly, we consider the situation of general finite metric graphs and prove upper estimates for all eigenvalues. These bounds turn out to be sharp; more precisely, they are attained by the equilateral figure-8 graph, i.e.\ the graph on one vertex and two loop edges attached to it. Remarkably, it can be shown that $\lambda_1 (\Gamma) \leq -1$ holds for any metric graph $\Gamma$, and equality is attained by any figure-8 graph, independent of its total length; this is a manifestation of the fact that the operator under consideration does not scale with a power of the graph size.

The proofs in this paper are based on so-called surgery principles, which we establish for the vertex conditions \eqref{eq:conditionsIntro}. It has been observed earlier for more standard vertex conditions that studying the effect of geometric perturbations of a metric graph on the eigenvalues is a powerful tool for proving spectral estimates, see, e.g., \cite{BKKM19,KMN13,KN14,RS20}. Such geometric perturbations include for instance merging or splitting vertices or attaching additional edges. It may be considered interesting on its own that the behavior of eigenvalues under surgery operations for the conditions \eqref{eq:conditionsIntro} turns out to differ partially from the behavior we know for continuity-Kirchhoff or $\delta$ coupling conditions. For instance, merging two vertices into one can increase or decrease the eigenvalues, depending on the parity of the involved vertex degrees. In the case of star graphs, some surgery principles are proven based on semi-explicit computations.

Let us describe the contents of the paper. In the following section we consider star graphs, compute the quadratic form corresponding to the conditions \eqref{eq:conditionsIntro}, and establish some elementary spectral properties. In the subsequent Section \ref{sec:groundState} the main optimization result for star graphs is established.
% the simplest graph of the form of a star; it edges meet in a single vertex and Neumann conditions are imposed at their opposite ends.
% Using the quadratic form associated with the corresponding Laplacian and an appropriate graph surgery, we show that the ground state
In Section~\ref{s: surgery} we develop the surgery idea further and prove upper bounds for higher eigenvalues on star graphs. In Section~\ref{s:Dirichlet} we briefly discuss how some spectral properties modify if we impose Dirichlet conditions at the `loose' edge endpoints, instead of Neumann. Finally, Section~\ref{sec:generalGraphs} contains the main result on optimal upper bounds on eigenvalues for a general compact graph with the vertex coupling \eqref{eq:conditionsIntro}.

%%%%%%%%%%%%%%%%%%%%%%%%%%%%%%%%%%%%%%%%%%%%%%%%%%%%%%%%%%%
\section{Vertex conditions, quadratic form, and elementary spectral properties}

In this section we introduce the non-standard vertex conditions under consideration and compute the corresponding quadratic form, which will be an important tool in our analysis. Moreover, we derive some elementary spectral properties. Since the vertex coupling is the essential ingredient of our discussion, here and in the next few sections we focus on compact \emph{star graphs}. The case of more general compact metric graphs will be the subject of Section \ref{sec:generalGraphs}.

To fix the present setting, $\Gamma$ will be a compact metric star graph with a single vertex $v_0$ of degree $N$ and edges $e_1, \dots, e_N$, identified with intervals $[0, l_j]$ such that the endpoint zero corresponds to the vertex $v_0$. For $k \in \N$ we write
 %----------------%
\begin{align*}
 \widetilde H^k (\Gamma) = \left\{ f : f_j \in H^k (0, l_j), j = 1, \dots, N \right\},
\end{align*}
 %----------------%
where $f_j$ denotes the restriction of the function $f : \Gamma \to \C$ to the edge $e_j \mathrel{\widehat{=}} [0, l_j]$. The object of our interest is the operator acting as negative second derivative on each edge with vertex conditions
 %----------------%
\begin{align}\label{eq:conditions}
 \left( F_{j + 1} - F_j \right) + i \left( F_j' + F_{j+1}' \right) = 0, \quad j = 1, \dots, N,
\end{align}
 %----------------%
where, for $f \in \widetilde H^2 (\Gamma)$,
 %----------------%
\begin{align*}
 F = (f_1 (0), f_2 (0), \dots, f_N (0))^\top \quad \text{and} \quad F' = (f_1' (0), f_2' (0), \dots, f_N' (0))^\top
\end{align*}
 %----------------%
with the entries being left limits; the indices are understood modulo $N$, i.e.\ $F_{N + 1}$ is identified with $F_1$ and so on.
% The case $N=2$ is trivial, without loss of generality we we may suppose that $N\ge 3$.
At the vertices of degree one, we impose Neumann boundary conditions. We point out that Neumann conditions formally correspond to the conditions \eqref{eq:conditions} at such vertices.
 %----------------%
\begin{remark}
An implication of the conditions \eqref{eq:conditions} is
\begin{align}\label{eq:Kirchhoff}
 \sum_{j = 1}^N F_j' = 0,
\end{align}
obtained by summation over $N$. Furthermore, if $N$ is even, then summing up the conditions multiplied by $(-1)^j$ leads to the implication
\begin{align}\label{eq:antiKirchhoff}
\sum_{j = 1}^N (-1)^j F_j = 0.
\end{align}
In particular, if $N = 2$, then these conditions are identic with the usual continuity-Kirchhoff vertex conditions, and the operator under consideration on a 2-star reduces to the Neumann Laplacian on the interval of length $l_1 + l_2$.
Note also that sometimes a variant of the conditions \eqref{eq:conditions} is considered with a factor $\ell>0$ in front of the second term at the left-hand side which fixes the length scale. The limit $\ell\to 0$ then returns us to continuity-Kirchhoff conditions.
\end{remark}
 %----------------%
\begin{remark}
On a graph with continuity-Kirchhoff conditions the Laplacian scales quadratically with the graph size. For the conditions \eqref{eq:conditions} this is not the case, and the reason is that, in the terminology of \cite{BK13}, the vertex coupling has a nontrivial Robin part. This happens for many other matching conditions. For instance, the $\delta$ coupling, in which the functions are continuous at the vertex and $\sum_{j = 1}^N F_j'$ is a (nonzero) multiple of the common value, has a one-dimensional Robin component; in case of condition \eqref{eq:conditions} the dimension of the Robin component is $\lfloor \frac{N}{2}\rfloor$.
\end{remark}
 %----------------%

In the following we find the expression of the quadratic form corresponding to the vertex conditions \eqref{eq:conditions}. For details on semi-bounded, closed sesquilinear forms and associated self-adjoint operators we refer to \cite[Chapter VI]{K95}.

\begin{proposition}\label{prop:form}
Let $N \geq 2$. The sesquilinear form
 %----------------%
\begin{align*}
 \sa [f, g] & = \sum_{j = 1}^N \int_0^{l_j} f_j' \overline{g_j'}{\,\dd x} + i \sum_{j = 2}^N \sum_{k = 1}^{j - 1} (-1)^{j + k} \left( F_k \overline{G_j} - F_j \overline{G_k} \right)
\end{align*}
 %----------------%
with domain
 %----------------%
\begin{align*}
 \dom \sa = \begin{cases}
             \widetilde H^1 (\Gamma), & \text{if}~N~\text{is odd},\\
             \bigg\{ f \in \widetilde H^1 (\Gamma) : \sum_{j = 1}^N (-1)^j F_j = 0 \bigg\},& \text{if}~N~\text{is even},
            \end{cases}
\end{align*}
 %----------------%
is densely defined, symmetric, semi-bounded below and closed. The corresponding self-adjoint operator in $L^2 (\Gamma)$ acts as the negative second derivative subject to the vertex conditions \eqref{eq:conditions} at $v_0$ and Neumann boundary conditions at the vertices of degree one.
\end{proposition}
 %----------------%
\begin{proof}
For any $N$, all smooth functions that vanish at $v_0$ belong to $\dom \sa$, hence the latter is dense in $L^2 (\Gamma)$. For the symmetry
% \begin{align*}
%  2 i \Imag \sum_{j = 1}^N (F_j - F_{j-1}) \overline{G_j} & = \sum_{j = 1}^N \left( (F_j - F_{j-1}) \overline{G_j} + \overline{(F_j - F_{j-1}) \overline{G_j}} \right)
% \end{align*}
 %----------------%
of $\sa$ it is sufficient to note that $\sa [f] := \sa [f, f]$ is real since
 %----------------%
\begin{align*}
 F_k \overline{F_j} - F_j \overline{F_k} & = 2 i \Imag \left( F_k \overline{F_j} \right)
\end{align*}
 %----------------%
holds for all $j$ and $k$. As for the semi-boundedness, if we choose $l > 0$ such that $l \leq l_j$ for $j = 1, \dots, N$, then by \cite[Lemma~1.3.8]{BK13}
 %----------------%
\begin{align*}
 |F_j|^2 \leq \frac{2}{l} \| f_j \|^2 + l \|f_j'\|^2
\end{align*}
 %----------------%
holds for $j = 1, \dots, N$, where $\| f_j \|$ denotes the norm of $f_j$ in $L^2 (0, l_j)$, etc. Consequently, for $j = 1, \dots, N$ we have
 %----------------%
\begin{align*}
 2 \Imag \left( F_k \overline{F_j} \right) & \leq 2 |F_k| |F_j| \leq \left( |F_k|^2 + |F_j|^2 \right) \\
 & \leq \frac{2}{l} \left( \|f_k\|^2 + \|f_{j}\|^2 \right) + l \left( \|f_k'\|^2 + \|f_{j}'\|^2 \right),
\end{align*}
 %----------------%
which implies
 %----------------%
\begin{align*}
 2 \Imag \sum_{j = 2}^N \sum_{k = 1}^{j - 1} (-1)^{j + k} \left( F_k \overline{F_j}\right) & \leq \frac{4 N}{l} \|f\|_{L^2 (\Gamma)}^2 + 2 N l \sum_{j = 1}^N \int_0^{l_j} |f_j'|^2 {\,\dd x}.
\end{align*}
 %----------------%
From this we obtain
 %----------------%
\begin{align}\label{eq:formEquivalence}
 \sa [f] & \geq \left(1 - 2 N l \right) \sum_{j = 1}^N \int_0^{l_j} |f_j'|^2{\,\dd x} - \frac{4 N}{l} \|f\|_{L^2 (\Gamma)}^2.
\end{align}
 %----------------%
Choosing $l \leq \frac{1}{2 N}$ we conclude that $\sa$ is semi-bounded below and that $- 4 N /l$ is a lower bound. Moreover, \eqref{eq:formEquivalence} indicates that the norm induced by $\sa$ on $\dom \sa$ is equivalent to the norm of $\widetilde H^1 (\Gamma)$. As $\dom \sa$ is closed in $\widetilde H^1 (\Gamma)$ it follows that $\sa$ is closed.

It remains to show that the self-adjoint operator $A$ associated with $\sa$ obeys the claimed vertex conditions. First of all, for $f \in \dom A$ and $g \in \dom \sa$ we have
 %----------------%
\begin{align*}
 (A f, g)_{L^2 (\Gamma)} & = \sa [f, g] = \sum_{j = 1}^N \int_0^{l_j} f_j' \overline{g_j'}{\,\dd x} + i \sum_{j = 2}^N \sum_{k = 1}^{j - 1} (-1)^{j + k} \left( F_k \overline{G_j} - F_j \overline{G_k} \right) \\
 & = - \sum_{j = 1}^N \int_0^{l_j} f_j'' \overline{g_j}{\,\dd x} + \sum_{j = 1}^N \left( f_j' (l_j) \overline{g_j (l_j)} - f_j' (0) \overline{g_j (0)} \right) \\
 & \quad + i \sum_{j = 2}^N \sum_{k = 1}^{j - 1} (-1)^{j + k} \left( F_k \overline{G_j} - F_j \overline{G_k} \right).
\end{align*}
 %----------------%
Choosing $g$ such that $g_j \in C_0^\infty (0, l_j)$ for $j = 1, \dots, N$ we obtain that $A$ acts as negative second derivative operator. In particular, we conclude that
 %----------------%
\begin{align}\label{eq:onTheWay}
 0 & = \sum_{j = 1}^N \left( f_j' (l_j) \overline{g_j (l_j)} - F_j' \overline{G_j} \right) + i \sum_{j = 2}^N \sum_{k = 1}^{j - 1} (-1)^{j + k} \left( F_k \overline{G_j} - F_j \overline{G_k} \right).
\end{align}
 %----------------%
If we choose for any given $j_0 \in \{1, \dots, N\}$ some $g$ such that $g_j \in C^\infty (0, l_j)$ with $g_j (0) = 0$ for $j = 1, \dots, N$, $g_j (l_j) = 0$ for $j \neq j_0$ and $g_{j_0} (l_{j_0}) = 1$ we see that $f$ satisfies a Neumann condition at the vertex $v_{j_0}$ to which $e_{j_0}$ is incident. In particular, \eqref{eq:onTheWay} simplifies
 %----------------%
\begin{align*}
 0 & = - \sum_{j = 1}^N F_j' \overline{G_j} + i \sum_{j = 2}^N \sum_{k = 1}^{j - 1} (-1)^{j + k} \left( F_k \overline{G_j} - F_j \overline{G_k} \right) \\
 & = - \sum_{j = 1}^N F_j' \overline{G_j} + i \left( \sum_{j = 2}^N \sum_{k = 1}^{j-1} (-1)^{j + k} F_k \overline{G_j} - \sum_{k = 1}^{N - 1} \sum_{j = k + 1}^N (-1)^{j + k} F_j \overline{G_k} \right).
\end{align*}
 %----------------%
% If we choose $g$ such that $G_j = 1$ for all $j$, then it follows $\sum_{j = 1}^N F_j' = 0$.
Let us now choose $g$ such that $G_j = 1$ if $j \in \{j_0, j_0 + 1\}$ and $G_j = 0$ otherwise (again identifying $G_{N+1}$ with $G_1$). We obtain
 %----------------%
\begin{align*}
 0 & = - F_{j_0}' - F_{j_0 + 1}' + i \Bigg( \sum_{k = 1}^{j_0-1} (-1)^{j_0 + k} F_k - \sum_{k = 1}^{j_0} (-1)^{j_0 + k} F_k \\
 & \quad - \sum_{j = j_0 + 1}^N (-1)^{j + j_0} F_j + \sum_{j = j_0 + 2}^N (-1)^{j + j_0} F_j \Bigg) \\
 & = - F_{j_0}' - F_{j_0 + 1}' + i \left( F_{j_0 + 1} - F_{j_0} \right),
\end{align*}
 %----------------%
(where we apply the convention that sums with lower index larger than upper index are zero) and the latter is valid for all $j_0 \in \{1, \dots, N\}$. These conditions are obviously equivalent to \eqref{eq:conditions}.

Conversely, let us assume that $f \in \widetilde H^2 (\Gamma)$ satisfies the conditions \eqref{eq:conditions} at the central vertex and Neumann conditions at the vertices of degree one. Our aim is to show $f \in \dom A$. Indeed, for each $g \in \dom \sa$ we have
\begin{align}\label{eq:formCompBackwards}
 \sa [f, g] & = - \sum_{j = 1}^N \int_0^{l_j} f_j'' \overline{g_j}{\,\dd x} - \sum_{j = 1}^N F_j' \overline{G_j} + i \sum_{j = 2}^N \sum_{k = 1}^{j - 1} (-1)^{j + k} \left( F_k \overline{G_j} - F_j \overline{G_k} \right),
\end{align}
where the Neumann vertex conditions at the vertices of degree one were used. Now we distinguish cases: if $N$ is odd, then we can use \eqref{eq:Kirchhoff} and then \eqref{eq:conditions} to obtain
\begin{align*}
 \sum_{j = 1}^N F_j' \overline{G_j} & = - \sum_{j = 1}^N \sum_{k \neq j} F_k' \overline{G_j} = i \sum_{j = 1}^N \left( \sum_{k = 1}^{j-1} (-1)^{j+k} F_k \overline{G_j} - \sum_{k = j + 1}^N (-1)^{j+k} F_k \overline{G_j} \right).
\end{align*}
From this and \eqref{eq:formCompBackwards}, $f \in \dom A$ follows if $N$ is odd. For the case of even $N$, note that successive application of \eqref{eq:conditions} leads to the identity
\begin{align*}
 F_j' = (-1)^{j + 1} F_1' + i \left( F_j + 2 \sum_{k = 2}^{j-1} (-1)^{j+k} F_k + (-1)^{j+1} F_1 \right), \quad j = 2, 3, \dots, N.
\end{align*}
With this we get
\begin{align*}
 \sum_{j = 1}^N F_j' \overline{G_j} & = F_1' \overline{G_1} + \sum_{j = 2}^N (-1)^{j+1} F_1' \overline{G_j} \\
 & \quad + i \sum_{j = 2}^{N} \left( F_j + 2 \sum_{k = 2}^{j-1} (-1)^{j+k} F_k + (-1)^{j+1} F_1 \right) \overline{G_j} \\
 & = i \sum_{j = 2}^N \left( \sum_{k = 2}^{j-1} (-1)^{j+k} F_k - \sum_{k = j+1}^N (-1)^{j+k} F_k \right) \overline{G_j},
\end{align*}
where the last step follows from \eqref{eq:antiKirchhoff}. Similarly to the above computations this can be rewritten
\begin{align*}
 \sum_{j = 1}^N F_j' \overline{G_j} & = i \sum_{j = 2}^N \sum_{k = 2}^{j-1} (-1)^{j+k} \left(F_k \overline{G_j} - F_j \overline{G_k} \right).
\end{align*}
Finally, the summands with $k = 1$ can be added without changing the value of the sum, since 
\begin{align*}
 \sum_{j = 2}^N (-1)^{j+1} \left( F_1 \overline{G_j} - F_j \overline{G_1} \right) = - F_1 \overline{G_1} - \left(- F_1 \overline{G_1} \right) = 0,
\end{align*}
again by \eqref{eq:antiKirchhoff}. Hence $f \in \dom A$ follows for even $N$ as well and the proof is complete.
% 
% An analogous reasoning shows that whenever $f \in \widetilde H^2 (\Gamma)$ satisfies the conditions \eqref{eq:conditions} \jr{at the central vertex and Neumann conditions at the vertices of degree one}, then $f$ belongs to $\dom A$.
\end{proof}

Due to compactness of $\Gamma$, the spectrum of the operator $A$ is discrete accumulating at infinity. Let us ask about the existence of negative eigenvalues. As we have said, the coupling is trivial for $N=2$ and $A$ is in this case unitarily equivalent to the negative second derivative operator with Neumann boundary conditions on an interval of length $l_1 + l_2$, so the negative spectrum is empty. For any $N \geq 3$, however, negative eigenvalues exist. A simple, direct argument for this is given in the following proposition; see Remark \ref{rem:inf_bound} below for a stronger statement.

\begin{proposition} \label{p:neg_exist}
Let $N \geq 3$. Then the operator under consideration has at least one negative eigenvalue.
\end{proposition}
 %----------------%
\begin{proof}
We have to distinguish between odd and even $N$. Let us first assume that $N$ is odd. Define an edgewise constant function $f$ by $f_1 = 1$, $f_2 = i$, and $f_j = 0$ identically otherwise. Then $f \in \dom \sa$ and
 %----------------%
\begin{align*}
 \sa [f] & = - 2 \Imag \sum_{j = 2}^2 \sum_{k = 1}^{j - 1} (-1)^{j + k} F_k \overline{F_j} = - 2 < 0.
\end{align*}
 %----------------%
Thus a negative eigenvalue must exist.

Now consider the case of an even $N \geq 4$. Define an edgewise constant function with values $f_1 = 1, f_2 = i, f_3 = - 1, f_4 = - i$ and $f_j = 0$ identically for all further edges (if any). Then $f$ satisfies $\sum_{j = 1}^N (-1)^j F_j = 0$ and thus it belongs to $\dom \sa$. Furthermore, a simple computation gives
 %----------------%
\begin{align*}
 \sa [f] & = - 2 \Imag \sum_{j = 2}^4 \sum_{k = 1}^{j - 1} (-1)^{j + k} F_k \overline{F_j} = - 4 < 0,
\end{align*}
 %----------------%
implying the existence of a negative eigenvalue.
\end{proof}

 %----------------%
\begin{remark} \label{rem:inf_bound}
We can actually obtain a stronger lower bound on the number of negative eigenvalues. By a simple bracketing argument \cite[Sec.~XIII.15]{RS78} it cannot be less than the number of negative eigenvalues of the star graph with same vertex degree and semiinfinite edges, $l_j=\infty,\, j=1,\dots,N$. This number is known for any self-adjoint vertex coupling \cite[Thm.~2.6]{BET22}. In case of \eqref{eq:conditions} we know even the eigenvalues explicitly \cite{ET18}; their number is $\lfloor \frac{N-1}{2}\rfloor$.
\end{remark}
 %----------------%

For finite star graphs an equation for the negative eigenvalues can be computed as follows. These computations will be useful in the context of the optimization problems considered below.
 %----------------%
\begin{example}\label{ex:computation}
Putting conventionally $\lambda = - \kappa^2$ with $\kappa>0$, we may use the ansatz
 %----------------%
\begin{align*}
 \psi_j (x) = \alpha_j \cosh (\kappa (l_j - x)).
\end{align*}
 %----------------%
Then
 %----------------%
\begin{align*}
 \psi_j' (x) = - \kappa \alpha_j \sinh(\kappa (l_j - x))
\end{align*}
 %----------------%
and, in particular, the Neumann boundary condition at the vertices of degree one is automatically satisfied. Furthermore, the vertex conditions at the central vertex translate into
 %----------------%
\begin{align*}
 \alpha_{j + 1} \cosh (\kappa l_{j+1}) - \alpha_j \cosh (\kappa l_j) - i \kappa \big( \alpha_j \sinh (\kappa l_j) + \alpha_{j+1} \sinh (\kappa l_{j+1}) \big) = 0
\end{align*}
 %----------------%
for $j = 1, \dots, N$. For abbreviation, let us set $A_j = \cosh (\kappa l_j) + i \kappa \sinh (\kappa l_j)$ and $B_j = - \cosh (\kappa l_j) + i \kappa \sinh (\kappa l_j)$. Then the vertex conditions are equivalent to
 %----------------%
\begin{align}\label{eq:LGS}
 \begin{pmatrix} A_1 & B_2 & 0 & \hdots & \hdots & 0 \\ 0 & \ddots & \ddots & \ddots &  & \vdots \\ \vdots & \ddots & \ddots & \ddots & \ddots & \vdots \\ \vdots & & \ddots & \ddots & \ddots & 0 \\ 0 & \hdots & \hdots & 0 & A_{N-1} & B_N \\ B_1 & 0 & \hdots & \hdots & 0 & A_N \end{pmatrix} \begin{pmatrix} \alpha_1 \\ \vdots \\ \alpha_N \end{pmatrix} = 0.
\end{align}
 %----------------%
By computing the determinant of the system's matrix we find that $\lambda = - \kappa^2$ is a negative eigenvalue if and only if
 %----------------%
\begin{align*}
 A_1 \cdots A_N + (-1)^{N+1} B_1 \cdots B_N = 0.
\end{align*}
 %----------------%
If $N = 3$, then $\lambda = - \kappa^2$ is a negative eigenvalue if and only if
 %----------------%
\begin{align}\label{eq:secularEquationN=3}
 \coth (\kappa l_1) \coth (\kappa l_2) + \coth (\kappa l_1) \coth (\kappa l_3) + \coth (\kappa l_2) \coth (\kappa l_3) = \kappa^2
\end{align}
 %----------------%
holds. For $N \geq 4$ the explicit form of the condition becomes more complicated. In the particular case when the star graph is equilateral with edge lengths $l$, then $\lambda = - \kappa^2$ is a negative eigenvalue if and only if
 %----------------%
\begin{align}\label{eq:secularEquationEquilateral}
 (\cosh (\kappa l) + i \kappa \sinh (\kappa l))^N + (-1)^{N+1} (- \cosh (\kappa l) + i \kappa \sinh (\kappa l))^N = 0.
\end{align}
 %----------------%

Let us now consider a few explicit choices for $N$, assuming furthermore that the star graph is equilateral. If $N = 3$, the equation \eqref{eq:secularEquationEquilateral} is equivalent to
 %----------------%
\begin{align}\label{eq:implicit}
 3 \coth^2 (\kappa l) = \kappa^2.
\end{align}
 %----------------%
In accordance with Proposition \ref{p:neg_exist} and Remark~\ref{rem:inf_bound} a solution must exist; it is straightforward to check that it is the case and the solution is unique for any $l>0$. If $l = 1$, for instance, we have $\kappa \approx 1.82$ corresponding to the negative eigenvalue $\lambda \approx - 3.33$. The coefficients of the corresponding eigenfunction can be computed, for any $l>0$, as follows. If $\kappa$ is the solution to \eqref{eq:implicit}, then
 %----------------%
\begin{align*}
 \cosh (\kappa l) = \frac{\kappa/\sqrt{3}}{\sqrt{\kappa^2/3 - 1}} \quad \text{and} \quad \sinh (\kappa l) = \frac{1}{\sqrt{\kappa^2/3 - 1}},
\end{align*}
 %----------------%
yielding
 %----------------%
\begin{align*}
 A = \frac{\kappa}{\sqrt{\kappa^2/3 - 1}} \left(1/\sqrt{3} + i \right) \quad \text{and} \quad B = \frac{\kappa}{\sqrt{\kappa^2/3 - 1}} \left(- 1/\sqrt{3} + i \right).
\end{align*}
 %----------------%
Now, returning to the system \eqref{eq:LGS} it is easy to see that
 %----------------%
\begin{align*}
 \alpha_1 = \ee^{\frac{2 \pi i}{3}}, \quad \alpha_2 = \ee^{\frac{4 \pi i}{3}}, \quad \alpha_3 = 1
\end{align*}
 %----------------%
is a solution and that it is unique up to multiples. In other words, there is a unique simple eigenvalue for any $l>0$, and the corresponding eigenfunction is supported on the entire graph. Moreover, it is clear from \eqref{eq:implicit} that this eigenvalue converges from below to $-3$, the eigenvalue of the infinite star graph \cite{ET18}, as $l\to\infty$.
% \marginpar{\tiny JR: it seems that the solution $\kappa$ converges to approx.\ $1.73$ as $l \to \infty$. Seems to be $\tan \pi/3$, the lowest eigenvalue of the infinite $3$-star. Can one show this convergence in general?}

For $N = 4$ one proceeds in a similar way. The secular equation acquires the form
 %----------------%
\begin{align*}
 \coth^2 (\kappa l) = \kappa^2
\end{align*}
 %----------------%
giving rise to a unique negative eigenvalue $\lambda = - \kappa^2$ is unique for any $l>0$ which, in addition, converges from below to $-1$ for $l\to\infty$ as expected \cite{ET18}. The coefficients of its corresponding eigenfunction can be chosen
 %----------------%
\begin{align*}
 \alpha_1 = i, \quad \alpha_2 = - 1, \quad \alpha_3 = - i, \quad \alpha_4 = 1.
\end{align*}
 %----------------%
For $l = 1$ one gets $\kappa \approx 1.20$. To add one more case, consider $N = 6$. Here the secular equation can be written as
 %----------------%
\begin{align*}
 3 \coth^4 (\kappa l) - 10 \coth^2 (\kappa l) \kappa^2 + 3 \kappa^4 = 0.
\end{align*}
 %----------------%
By solving a quadratic equation we obtain
 %----------------%
\begin{align*}
 \coth^2 (\kappa l) = \Big( \frac{5}{3} \pm \frac{4}{3} \Big) \kappa^2,
\end{align*}
 %----------------%
leading to the two distinct equations
 %----------------%
\begin{align*}
 \coth (\kappa l) = \frac{\kappa}{\sqrt{3}} \quad \text{and} \quad \coth (\kappa l) = \sqrt{3} \kappa
\end{align*}
 %----------------%
and, thus, to two distinct negative eigenvalues. In particular, for $l = 1$ the solutions are $\kappa \approx 1.82$ and $\kappa \approx 0.84$, respectively, and in the limit $l\to\infty$ the eigenvalues tend to $3$ and $\frac13$, respectively, as expected \cite{ET18}. For the first equation we can compute the coefficients of the eigenfunction to be
 %----------------%
\begin{align*}
 \alpha_1 = \alpha_4 = \ee^{\frac{4 \pi i}{3}}, \quad \alpha_2 = \alpha_5 = \ee^{\frac{2 \pi i}{3}}, \quad \alpha_3 = \alpha_6 = 1,
\end{align*}
 %----------------%
while for the second equation we get $\alpha_j = \ee^{\frac{j \pi i}{3}},\, j = 1, \dots, 6$, hence again the eigenfunction components on different edges differ by coefficients which are roots of unity. Note also that in all three cases the lower bound of Remark~\ref{rem:inf_bound} is saturated.

\end{example}

% \begin{example}
% For $N = 3$ the computations of the previous example yield the secular equation
% \begin{align*}
%  \coth (\kappa l_1) \coth (\kappa l_3) + \coth (\kappa l_2) \coth (\kappa l_3) + \coth (\kappa l_1) \coth (\kappa l_2) - \kappa^2 = 0.
% \end{align*}
% {\bf Can one compute explicitly which lengths maximize the (lowest) negative eigenvalue when fixing the total length? By the surgery principles of the next section, this would yield a bound for all $N$.}
% \end{example}

%%%%%%%%%%%%%%%%%%%%%%%%%%%%%%%%%%%%%%%%%%%%%%%%%%%
\section{Optimizing the ground state eigenvalue of a star graph}\label{sec:groundState}

In this section we study the behavior of the lowest eigenvalue on a star graph when the edge lengths are varied but the total length is fixed, in other words, when the volume is preserved but its parts are transplanted between edges. We denote by $\lambda_1 (\Gamma) < 0$ the first eigenvalue of the Laplacian on the metric star graph $\Gamma$ with matching conditions \eqref{eq:conditions} at the central vertex and Neumann boundary conditions at the `loose ends', i.e. the vertices of degree one.
% As usual we denote by $e_1, \dots, e_N$ the edges of $\Gamma$ and by $l_j$ the length of $e_j$, $j = 1, \dots, N$.
 %----------------%
\begin{definition}
Let $\Gamma$ be a metric star graph with $N$ edges $e_1, \dots, e_N$ of lengths $l_1, \dots, l_N$. We say that {\em $\widetilde \Gamma$ is obtained from $\Gamma$ by transplanting a segment of length $\ell \leq l_j$ from the edge $e_j$ to the edge $e_k$} if $\widetilde \Gamma$ is the metric star graph with $N$ edges of lengths $\widetilde l_1, \dots, \widetilde l_N$, where $\widetilde l_j = l_j - \ell$, $\widetilde l_k = l_k + \ell$ and $\widetilde l_n = l_n$ whenever $n \neq j, k$.
\end{definition}
 %----------------%
In the previous definition the case $\ell = l_j$ is formally admitted. In this case, $\widetilde l_j = 0$ and we identify $\widetilde \Gamma$ with a star graph of $N - 1$ edges.
 %----------------%
\begin{proposition}\label{prop:transplantation}
Let $\Gamma$ be a star graph with $N \geq 3$ edges. Assume that the star graph $\widetilde \Gamma$ is obtained from $\Gamma$ by transplanting a segment of length $\ell < l_j$ from the edge $e_j$ to the edge $e_k$. Then
 %----------------%
\begin{align}\label{eq:transplantation}
 \lambda_1 (\widetilde \Gamma) < \lambda_1 (\Gamma)
\end{align}
 %----------------%
holds if $l_j \leq l_k$. The same assertion remains true if $\ell = l_j$ and $N$ is even.
\end{proposition}
 %----------------%
\begin{proof}
Without loss of generality let us assume that the transplantation is done from the edge $e_1$ to the edge $e_2$ and that $l_1 \leq l_2$. We first assume that $\ell < l_1$; the case $\ell = l_1$ will be discussed below. Let $\psi$ be an eigenfunction on $\Gamma$ corresponding to the eigenvalue $\lambda_1 (\Gamma)$ so that, in particular,
 %----------------%
\begin{align*}
 \sa [\psi] = \lambda_1 (\Gamma) \|\psi\|_{L^2 (\Gamma)}^2.
\end{align*}
 %----------------%
From Example \ref{ex:computation} we know that its component on the edge $e_j$ is given by
 %----------------%
\begin{align*}
 \psi_j (x) = \alpha_j \cosh (\kappa (l_j - x)), \quad 0 \leq x \leq l_j,
\end{align*}
 %----------------%
where $\kappa$ is positive with $\lambda_1 (\Gamma) = - \kappa^2$ and the complex coefficients $\alpha_j$ satisfy
 %----------------%
\begin{align}\label{eq:coefficientRelations}
 A_j \alpha_j + B_{j+1} \alpha_{j+1} = 0, \quad j = 1, \dots, N,
\end{align}
 %----------------%
using the `cyclic' convention identifying $A_{N+1}$ with $A_1$, etc. Recall that the coefficients are given by
 %----------------%
\begin{align*}
 A_j = \cosh (\kappa l_j) + i \kappa \sinh (\kappa l_j)
\end{align*}
 %----------------%
and
 %----------------%
\begin{align*}
 B_j = - \cosh (\kappa l_j) + i \kappa \sinh (\kappa l_j)
\end{align*}
 %----------------%
for $j = 1, \dots, N$. This means, in particular, that
 %----------------%
\begin{align*}
 |A_j|^2 = |B_j|^2 = \cosh^2 (\kappa l_j) + \kappa^2 \sinh^2 (\kappa l_j).
\end{align*}
 %----------------%
Since
 %----------------%
\begin{align*}
 \frac{\textup{d}}{\textup{d} x} \left(\cosh^2 (\kappa x) + \kappa^2 \sinh^2 (\kappa x)\right) = 2 \kappa \cosh (\kappa x) \sinh (\kappa x) (1 + \kappa^2) > 0
\end{align*}
 %----------------%
holds for $x > 0$, we infer that $|A_j|$ and $|B_j|$ are strictly increasing as functions of the edge length $l_j$. This means that
 %----------------%
\begin{align*}
 |B_2| = |A_2| \geq |A_1|
\end{align*}
 %----------------%
holds if $l_1 \leq l_2$. In combination with \eqref{eq:coefficientRelations} this gives
 %----------------%
\begin{align}\label{eq:orderRelation}
 1 \geq \frac{|A_1|}{|B_2|} = \frac{|\alpha_2|}{|\alpha_1|}.
\end{align}
 %----------------%

To prove the inequality \eqref{eq:transplantation} we construct a suitable trial function for the quadratic form $\widetilde \sa$, the counterpart of $\sa$ for $\widetilde \Gamma$. Similarly as for $\Gamma$, we assume that the edges of $\widetilde \Gamma$ are parametrized by the intervals $[0, \widetilde l_j]$ with the zero endpoint corresponding to the central vertex. We define a trial function $\widetilde \psi$ on $\widetilde \Gamma$ by setting
 %----------------%
\begin{align*}
 & \widetilde \psi_1 (x) = \psi_1 (x), \quad x \in [0, \widetilde l_1], \\
 & \widetilde \psi_2 (x) = \begin{cases} \psi_2 (x), & x \in [0, l_2], \\ \frac{\alpha_2}{\alpha_1 \cosh (\kappa \ell)} \psi_1 (x - l_2 + \widetilde l_1), & x \in [l_2, \widetilde l_2], \end{cases},
\end{align*}
 %----------------%
together with $\widetilde \psi_j = \psi_j$ for $j = 3, \dots, N$. That is, we transplant a piece of the eigenfunction from $e_1$ to $e_2$, multiplied by some scaling factor. Then we have $\widetilde \Psi_j := \widetilde \psi_j (0) = \psi_j (0) =: \Psi_j$, $j = 1, \dots, N$, and $\widetilde \psi_j \in H^1 (0, \widetilde l_j)$ for $j \neq 2$. Furthermore, $\widetilde \psi_2$ is well-defined and continuous at $x = l_2$ since
 %----------------%
\begin{align*}
  \frac{\alpha_2}{\alpha_1 \cosh (\kappa \ell)} \psi_1 (l_2 - l_2 + \widetilde l_1) = \alpha_2 = \psi_2 (l_2).
\end{align*}
 %----------------%
Thus  $\widetilde \psi_2 \in H^1 (0, \widetilde l_2)$ as well, and consequently, $\widetilde \psi \in \dom \widetilde \sa$. Furthermore, using \eqref{eq:orderRelation} and the property $\widetilde \psi_j (0) = \psi_j (0)$ again for $j = 1, \dots, N$, we infer that
 %----------------%
\begin{align*}
 \widetilde \sa [\widetilde \psi] & = \int_0^{l_1 - \ell} |\psi_1'|^2 {\,\dd x} + \int_0^{l_2} |\psi_2'|^2 {\,\dd x} + \left| \frac{\alpha_2}{\alpha_1 \cosh (\kappa \ell)} \right|^2 \int_{l_2}^{l_2 + \ell} |\psi_1' (x - l_2 + \widetilde l_1)|^2 {\,\dd x} \\
 & \quad + \sum_{j = 3}^N \int_0^{l_j} |\psi_j'|^2 {\,\dd x} + \textup{vertex terms} \\
 & \leq \sa [\psi] = \lambda_1 (\Gamma) \|\psi\|_{L^2 (\Gamma)}^2
\end{align*}
 %----------------%
using $\cosh (\kappa \ell) > 1$; note that the vertex terms of $\widetilde \sa [\widetilde \psi]$ and $\sa [\psi]$ coincide. On the other hand, using \eqref{eq:orderRelation} once more, we get
 %----------------%
\begin{align*}
 \|\psi\|^2 & \geq \int_0^{l_1 - \ell} |\psi_1|^2 {\,\dd x} + \int_0^{l_2} |\psi_2|^2 {\,\dd x} + \left| \frac{\alpha_2}{\alpha_1 \cosh (\kappa \ell)} \right|^2 \int_{l_2}^{l_2 + \ell} |\psi_1 (x - l_2 + \widetilde l_1)|^2 {\,\dd x} \\
 & \quad + \sum_{j = 3}^N \int_0^{l_j} |\psi_j|^2 {\,\dd x} \\
 & = \|\widetilde \psi\|_{L^2 (\widetilde \Gamma)}^2.
\end{align*}
 %----------------%
As $\lambda_1 (\Gamma) < 0$, this implies
 %----------------%
\begin{align*}
 \widetilde \sa [\widetilde \psi] \leq \lambda_1 (\Gamma) \|\widetilde \psi\|_{L^2 (\widetilde \Gamma)}^2.
\end{align*}
 %----------------%
Since $\widetilde \psi$ cannot be an eigenfunction of $\widetilde \Gamma$ (because, for instance, it does not satisfy Neumann boundary conditions at the vertices of degree one corresponding to $e_1$ and $e_2$ in $\widetilde \Gamma$), it follows $\lambda_1 (\widetilde \Gamma) < \lambda_1 (\Gamma)$.

Let us briefly discuss the adjustments which need to be made to adapt this proof to the case $\ell = l_1$. We can construct the trial function $\widetilde \psi$ in the same way as above. Since we are assuming now that $N$ is even, the resulting graph $\widetilde \Gamma$ has an odd number of edges, and consequently, we still have $\widetilde \psi \in \dom \widetilde \sa$. Furthermore, the original eigenfunction $\psi$ on $\Gamma$ satisfies
 %----------------%
\begin{align*}
 \sum_{j=1}^N (-1)^j \Psi_j = 0.
\end{align*}
 %----------------%
Hence the vertex terms of $\widetilde \sa [\widetilde \psi]$ can be written
 %----------------%
\begin{align*}
 - 2 \Imag \sum_{j = 3}^{N} \sum_{k = 2}^{j - 1} (-1)^{j + k} \left( \Psi_k \overline{\Psi_j} \right) & = - 2 \Imag \sum_{j = 2}^N \left( \sum_{k = 1}^{j-1} (-1)^{j + k} (\Psi_k \overline{\Psi_j}) - (-1)^{j+1} (\Psi_1 \overline{\Psi_j}) \right) \\
 & = - 2 \Imag \left( \sum_{j = 2}^N \sum_{k = 1}^{j-1} (-1)^{j + k} (\Psi_k \overline{\Psi_j}) + |\Psi_1|^2 \right) \\
 & = - 2 \Imag \sum_{j = 2}^N \sum_{k = 1}^{j-1} (-1)^{j + k} (\Psi_k \overline{\Psi_j}),
\end{align*}
 %----------------%
which is precisely the vertex terms of $\sa [\psi]$. Thus the proof still works when $\ell = l_1$ and $N$ is even.
\end{proof}

As a consequence, we obtain the following result.
 %----------------%
\begin{corollary}\label{cor:max}
Assume that $\Gamma$ is an arbitrary metric star graph with $N$ edges and $\Gamma^*$ is the equilateral star graph with $N$ edges and the same total length as $\Gamma$. Then
 %----------------%
\begin{align*}
 \lambda_1 (\Gamma) \leq \lambda_1 (\Gamma^*)
\end{align*}
 %----------------%
holds. Equality prevails if and only if $\Gamma$ is equilateral.
\end{corollary}
 %----------------%
\begin{proof}
According to Proposition \ref{prop:transplantation}, transplanting a piece of positive length of a longer edge to a shorter one increases the lowest eigenvalue strictly. Starting from $\Gamma$ we can indeed perform a sequence of transplantations of this type and eventually arrive at $\Gamma^*$. In fact, assume without loss of generality that the edges $e_1, \dots, e_N$ of $\Gamma$ are enumerated such that
 %----------------%
\begin{align*}
 l_1 \leq l_2 \leq \dots \leq l_K < \overline{l} \leq l_{K+1} \leq \dots \leq l_N
\end{align*}
 %----------------%
holds for some index $K$, where $\overline{l} = \frac{1}{N} \sum_{j = 1}^N l_j$ is the average edge length. We first choose the longest edge, $e_N$, and transplant subsequently $K$ pieces of the same length $(l_N - \overline l)/K$ from $e_N$ to each of the edges $e_1, \dots, e_K$. As each of these edges has initially been shorter than $\overline l$ and $e_N$ remains longer or equally long as $\overline l$ during each of these transplantations, the lowest eigenvalue increases in each of these steps. The result is a star graph $\Gamma'$ with $N$ edges and the same total length as $\Gamma$, in which the edge $e'_N$ has length $\overline l$. If $\Gamma'$ is equilateral, we are done. Otherwise we apply the same procedure as before to reduce the length of the longest edge of $\Gamma'$ to $\overline{l}$, thereby increasing the lowest eigenvalue. As every application of this procedure reduces one more edge to the average length, after finitely many steps we arrive at $\Gamma^*$, proving in this way that $\lambda_1 (\Gamma) < \lambda_1 (\Gamma^*)$ holds unless $\Gamma$ was equilateral from the beginning.
\end{proof}
 %----------------%

In the next step, we compare equilateral star graphs of fixed total length with varying number of edges.
 %----------------%
\begin{proposition}\label{prop:manyInequalities}
Let $L > 0$ be fixed. For each $n$ denote by $\Gamma_n^*$ the equilateral star graph with $n$ edges where each edge has length $L/n$, that is, the total length of $\Gamma_n^*$ is $L$. For each $N \geq 3$ the following assertions hold.
 %----------------%
\begin{enumerate}
 \item $\lambda_1 (\Gamma_{N+2}^*) < \lambda_1 (\Gamma_N^*)$;
 \item $\lambda_1 (\Gamma_N^*) < \lambda_1 (\Gamma_{N+1}^*)$ if $N$ is odd; and
 \item $\lambda_1 (\Gamma_{N+1}^*) < \lambda_1 (\Gamma_N^*)$ if $N$ is even.
\end{enumerate}
 %----------------%
\end{proposition}
 %----------------%
\begin{proof}
Throughout this proof we denote by $\sa_n$ the quadratic form of Proposition \ref{prop:form} on $\Gamma_n^*$.

(i) Let $\psi$ be an eigenfunction on $\Gamma_N^*$ associated with $\lambda_1 (\Gamma_N^*)$. On $\Gamma_{N+2}^*$ we define a trial function $\widetilde \psi$ by
 %----------------%
\begin{align*}
 \widetilde \psi_j (x) = \begin{cases}\psi_j (x), \quad x \in \left[0, \frac{L}{N + 2} \right], & j = 1, \dots, N, \\ 0, & j = N+1, N+2, \end{cases}
\end{align*}
 %----------------%
that is, we truncate the eigenfunction of $\Gamma_N^*$ and extend it by zero to the additional edges. Then $\widetilde \psi_j \in H^1 (0, \frac{L}{N+2})$ for each $j$, and in case $N$ is even, $\widetilde \psi$ inherits the vertex conditions at $v_0$ from $\psi$. Thus, in any case, $\widetilde \psi \in \dom \sa_{N+2}$. We compute $\sa_{N+2} [\widetilde \psi]$ for even $N$. In this case we have
 %----------------%
\begin{align*}
 \sa_{N+2} [\widetilde \psi] & = \sum_{j = 1}^{N+2} \int_0^{\frac{L}{N+2}} |\widetilde \psi_j'|^2 {\,\dd x} - 2 \sum_{j = 2}^{N+2} \sum_{k = 1}^{j - 1} (-1)^{j + k} \Imag \left( \widetilde \Psi_k \overline{\widetilde \Psi_j} \right) \\
 & < \sum_{j = 1}^{N} \int_0^{\frac{L}{N}} |\psi_j'|^2 {\,\dd x} - 2 \sum_{j = 2}^{N} \sum_{k = 1}^{j - 1} (-1)^{j + k} \Imag \left( \Psi_k \overline{\Psi_j} \right) \\
 & = \sa_N [\psi] = \lambda_1 (\Gamma_N^*) \|\psi\|_{L^2 (\Gamma)}^2 < \lambda_1 (\Gamma_N^*) \|\widetilde \psi\|_{L^2 (\widetilde \Gamma)}^2,
\end{align*}
 %----------------%
where in the last step we used the fact that $\|\psi\|_{L^2 (\Gamma)} > \|\widetilde \psi\|_{L^2 (\widetilde \Gamma)}$ and $\lambda_1 (\Gamma_N^*) < 0$. This implies (i).

(ii) This follows from Proposition \ref{prop:transplantation}: one proceeds by transplanting successively pieces of length $\frac{L}{N (N + 1)}$ from the edge $e_{N+1}$ of $\Gamma_{N+1}$ to the edges $e_j$, $j = 1, \dots, N$; in the last of these steps, when the length of $e_{N+1}$ is `reduced to zero' we use the fact that $N+1$ is even by assumption.

(iii) In analogy with the proof of (i), let $\psi$ be an eigenfunction on $\Gamma_N^*$ corresponding to $\lambda_1 (\Gamma_N^*)$ and define $\widetilde \psi$ on $\Gamma_{N+1}^*$ by
 %----------------%
\begin{align*}
 \widetilde \psi_j (x) = \begin{cases}\psi_j (x), \quad x \in \left[0, \frac{L}{N + 1} \right], & j = 1, \dots, N, \\ 0, & j = N+1. \end{cases}
\end{align*}
 %----------------%
Then $\widetilde \psi \in \dom \sa_{N+1}$ (no vertex conditions are required as $N+1$ is odd) and
 %----------------%
\begin{align*}
 \sa_{N+1} [\widetilde \psi] & = \sum_{j = 1}^{N+1} \int_0^{\frac{L}{N+1}} |\widetilde \psi_j'|^2 {\,\dd x} - 2 \sum_{j = 2}^{N+1} \sum_{k = 1}^{j-1} (-1)^{j + k} \Imag \left( \widetilde \Psi_k \overline{\widetilde \Psi_j} \right) \\
 & < \sum_{j = 1}^{N} \int_0^{\frac{L}{N}} |\psi_j'|^2 {\,\dd x} - 2 \sum_{j = 2}^{N} \sum_{k = 1}^{j-1} (-1)^{j + k} \Imag \left(\Psi_k \overline{\Psi_j} \right) \\
 & = \sa_N [\psi] = \lambda_1 (\Gamma_N^*) \|\psi\|^2 < \lambda_1 (\Gamma_N^*) \|\widetilde \psi\|^2,
\end{align*}
 %----------------%
where the inequalities $\|\psi\|_{L^2 (\Gamma)} > \|\widetilde \psi\|_{L^2 (\widetilde \Gamma)}$ and $\lambda_1 (\Gamma_N^*) < 0$ were used. This yields assertion (iii) and concludes the proof.
\end{proof}

With this preliminary, one can obtain easily the following optimization theorem for the ground state eigenvalue.
 %----------------%
\begin{theorem}\label{thm:optimizationGroundState}
Let $L > 0$ and let $\Gamma$ be a metric star graph with $N \geq 3$ edges and total length $L$. If we denote by $\Gamma_3^*$ and $\Gamma_4^*$ the equilateral star graphs with 3 and 4 edges, respectively, of total length $L$, then
 %----------------%
\begin{align*}
 \lambda_1 (\Gamma) \leq \begin{cases} \lambda_1 (\Gamma_3^*), & \text{if}~N~\text{is odd}, \\ \lambda_1 (\Gamma_4^*), & \text{if}~N~\text{is even}, \end{cases}
\end{align*}
 %----------------%
holds, with equality if and only if $\Gamma = \Gamma_3^*$ (if $N$ is odd) or $\Gamma = \Gamma_4^*$ (if $N$ is even).
\end{theorem}
 %----------------%
\begin{proof}
By Corollary \ref{cor:max} we have $\lambda_1 (\Gamma) \leq \lambda_1 (\Gamma_N^*)$ where $\Gamma_N^*$ is the equilateral star graph with $N$ edges and total length $L$, with equality if and only if $\Gamma = \Gamma_N^*$. Next we apply Proposition \ref{prop:manyInequalities} (i) to obtain
 %----------------%
\begin{align*}
 \lambda_1 (\Gamma_N^*) \leq \begin{cases} \lambda_1 (\Gamma_3^*), & \text{if}~N~\text{is odd}, \\ \lambda_1 (\Gamma_4^*), & \text{if}~N~\text{is even}, \end{cases}
\end{align*}
 %----------------%
the equality being valid if and only if $N = 3$ (if $N$ is odd) or $N = 4$ (if $N$ is even). The claim of the theorem follows.
\end{proof}

%%%%%%%%%%%%%%%%%%%%%%%%%%%%%%%%%%%%%%%%%%%%%%%%%%%
\section{Further surgery principles and spectral estimates on star graphs}
\label{s: surgery}

Using medical terminology again, we will go now beyond transplantations and prove some surgery principles in which volume is added to (respectively removed from) a star graph. We denote by
 %----------------%
\begin{align*}
 \lambda_1 (\Gamma) \leq \lambda_2 (\Gamma) \leq \dots
\end{align*}
 %----------------%
the eigenvalues of the Laplacian on a compact metric star graph $\Gamma$ with vertex conditions \eqref{eq:conditions} at the central vertex and Neumann boundary conditions otherwise, numerated with their multiplicities taken into account.
 %----------------%
\begin{proposition}\label{prop:addEdge}
Assume that $\Gamma$ is a star graph with $N \geq 2$ edges and that the star graph $\widetilde \Gamma$ is obtained from $\Gamma$ by adding one edge rooted in the common vertex. Then
 %----------------%
\begin{align*}
 \lambda_k (\widetilde \Gamma) \leq \lambda_k (\Gamma)
\end{align*}
 %----------------%
holds
\begin{enumerate}
 \item for all $k \in \N$ if $N$ is even; and
 \item for all $k \in \N$ such that $\lambda_k (\Gamma) \geq 0$ if $N$ is odd.
\end{enumerate}
 %----------------%
\end{proposition}
 %----------------%
\begin{proof}
(i) Let us first consider the case where $\Gamma$ has an even number of $N \geq 2$ edges and $\widetilde \Gamma$ is obtained from $\Gamma$ by adding an additional edge $e_{N+1}$. Let us denote by $\sa$ and $\widetilde \sa$ the quadratic forms for $\Gamma$ and $\widetilde \Gamma$, respectively. For any $f \in \dom \sa$ we denote by $g$ its extension by zero to $\widetilde \Gamma$. Then
 %----------------%
\begin{align*}
 \widetilde \sa [g] & = \sum_{j = 1}^{N+1} \int_0^{l_j} |g_j'|^2 {\,\dd x} - 2 \Imag \sum_{j = 2}^{N+1} \sum_{k = 1}^{j-1} (-1)^{j + k} \left( G_k \overline{G_j} \right) \\
 & = \sum_{j = 1}^{N} \int_0^{l_j} |f_j'|^2 {\,\dd x} - 2 \Imag \sum_{j = 2}^{N} \sum_{k = 1}^{j-1} (-1)^{j + k} \left( F_k \overline{F_j} \right) \\
 & = \sa [f].
\end{align*}
 %----------------%
Now let $\cF$ be a $k$-dimensional subspace of $\dom \sa$ such that
 %----------------%
\begin{align*}
 \sa [f] \leq \lambda_k (\Gamma) \|f\|_{L^2 (\Gamma)}^2, \quad f \in \cF.
\end{align*}
 %----------------%
Then the space $\cG$ of extensions $g$ of $f \in \cF$ by zero has dimension $k$ and we have
 %----------------%
\begin{align*}
 \widetilde \sa [g] = \sa [f] \leq \lambda_k (\Gamma) \|f\|_{L^2 (\Gamma)}^2 = \lambda_k (\Gamma) \|g\|_{L^2 (\widetilde \Gamma)}^2, \quad g \in \cG.
\end{align*}
 %----------------%
By min-max principle, $\lambda_k (\widetilde \Gamma) \leq \lambda_k (\Gamma)$ follows thus for all $k \in \N$ in the case that $\Gamma$ has an even number of edges.

(ii) Consider next the case that $N \geq 3$ is odd, $\Gamma$ has edges $e_1, \dots, e_N$ and $\widetilde \Gamma$ has an additional edge $e_{N+1}$. Denote again by $\sa$ and $\widetilde \sa$ the quadratic forms for $\Gamma$ and $\widetilde \Gamma$, respectively. For any $f \in \dom \sa$ let $g$ be its extension to $\widetilde \Gamma$ by the function with constant value $g_{N+1} = \sum_{j = 1}^N (-1)^{j+1} F_j$ on the additional edge $e_N$. Then it is easy to check that we have $g \in \dom \widetilde a$, and
 %----------------%
\begin{align*}
 \widetilde \sa [g] & = \sum_{j = 1}^{N+1} \int_0^{l_j} |g_j'|^2 {\,\dd x} - 2 \Imag \sum_{j = 2}^{N+1} \sum_{k = 1}^{j-1} (-1)^{j + k} \left(G_k \overline{G_j} \right) \\
 & = \sum_{j = 1}^N \int_0^{l_j} |f_j'|^2 {\,\dd x} - 2 \Imag \sum_{j = 2}^N \sum_{k = 1}^{j-1} (-1)^{j + k} \left(F_k \overline{F_j} \right) \\
 & \quad - 2 \Imag \sum_{k = 1}^N (-1)^{N + 1 + k} \bigg( F_k \sum_{j = 1}^N (-1)^{j+1} \overline{F_j} \bigg) \\
 & = \sa [f],
\end{align*}
 %----------------%
because
 %----------------%
\begin{align*}
 - \Imag \sum_{k = 1}^N (-1)^{N + 1 + k} \bigg( F_k \sum_{j = 1}^N (-1)^{j+1} \overline{F_j} \bigg) = \Imag |g_{N + 1}|^2 = 0.
\end{align*}
 %----------------%
Let now $\cF$ be a $k$-dimensional subspace of $\dom \sa$ such that
 %----------------%
\begin{align*}
 \sa [f] \leq \lambda_k (\Gamma) \|f\|_{L^2 (\Gamma)}^2, \quad f \in \cF.
\end{align*}
 %----------------%
Then the space $\cG$ of extensions $g$ of $f \in \cF$ constructed as described above has the dimension $k$; note that elements of $\cF$ are linearly independent if and only if the same is true for their extension to $\cG$. We have
 %----------------%
\begin{align*}
 \widetilde \sa [g] = \sa [f] \leq \lambda_k (\Gamma) \|f\|_{L^2 (\Gamma)}^2 \leq \lambda_k (\Gamma) \|g\|_{L^2 (\widetilde \Gamma)}^2, \quad g \in \cG,
\end{align*}
 %----------------%
as long as $\lambda_k (\Gamma) \geq 0$, and the assertion for odd $N$ thus follows.
\end{proof}

By analogous reasoning we obtain the following claim.

\begin{proposition}\label{prop:evenEdges}
Assume that $\Gamma$ is a star graph with an even number of edges and that $\widetilde \Gamma$ is obtained from $\Gamma$ by attaching any even number of additional edges. Then
 %----------------%
\begin{align*}
 \lambda_k (\widetilde \Gamma) \leq \lambda_k (\Gamma)
\end{align*}
 %----------------%
holds for all $k \in \N$.
\end{proposition}

 %----------------%
\begin{corollary}
Let $\Gamma$ be a star graph with an even number $N \geq 2$ of edges. Then
 %----------------%
\begin{align*}
 \lambda_k (\Gamma) \leq \frac{(k-1)^2 \pi^2}{\diam (\Gamma)^2} \leq \frac{(k-1)^2 \pi^2 N^2}{4 L (\Gamma)^2} = \frac{(k-1)^2 \pi^2}{4 A (\Gamma)^2}
\end{align*}
 %----------------%
holds for all $k \in \N$, where $L (\Gamma)$ denotes the total length of $\Gamma$, $A (\Gamma)$ the arithmetic mean of the edge lengths, and $\diam (\Gamma)$ the metric diameter of $\Gamma$.
\end{corollary}
 %----------------%
\noindent It should be remarked that this bound is trivial for $k = 1$ as it only yields $\lambda_1 (\Gamma) \leq 0$; recall that by Proposition~\ref{p:neg_exist} we know that $\lambda_1 (\Gamma) < 0$ holds, in particular, for $N\ge 4$, and this will be sharpened in Theorem \ref{thm:optimizationGeneral} below.
 %----------------%
\begin{proof}
Without loss of generality we may assume that the edges are numbered in such a way that $l_1 \geq l_2 \geq \dots \geq l_N$. Then $l_1 + l_2 = \diam (\Gamma)$ and
 %----------------%
\begin{align}\label{eq:pigeonhole}
 \frac{l_1 + l_2}{2} \geq \frac{L (\Gamma)}{N} = A (\Gamma).
\end{align}
 %----------------%
Let $\cP$ be the path graph obtained from $\Gamma$ by removing all the edges except $e_1$ and $e_2$, which is certainly an even number. What we get is a 2-star and the vertex conditions \eqref{eq:conditions} at the vertex simplify to Kirchhoff conditions, which in the case of a 2-star is trivial, i.e. the spectrum of $\cP$ coincides with the spectrum of the Neumann Laplacian on an interval of length $l_1 + l_2$. Proposition~\ref{prop:evenEdges} then gives
 %----------------%
\begin{align*}
 \lambda_k (\Gamma) \leq \lambda_k (\cP) = \frac{(k-1)^2 \pi^2}{L (\cP)^2} = \frac{(k-1)^2 \pi^2}{\diam (\Gamma)^2} = \frac{(k-1)^2 \pi^2}{(l_1 + l_2)^2} \leq \frac{(k-1)^2 \pi^2 N^2}{4 L (\Gamma)^2},
\end{align*}
 %----------------%
where \eqref{eq:pigeonhole} was used; this we have set out to prove.
\end{proof}

We have seen in Example~\ref{ex:computation} that in some cases the ground state eigenvalue increases with edge length. In fact, we have the following more general result.
 %----------------%
\begin{proposition} \label{p:extens}
Let $\Gamma$ be a star graph. Assume that $\widetilde \Gamma$ is obtained from $\Gamma$ by extending the length of an edge. Then, fixing a $k \in \N$, we get:
 %----------------%
\begin{enumerate}
 \item If $\lambda_k (\widetilde \Gamma) \leq 0$, then $\lambda_k (\Gamma) \leq \lambda_k (\widetilde \Gamma)$.
 \item If $\lambda_k (\Gamma) \geq 0$, then $\lambda_k (\widetilde \Gamma) \leq \lambda_k (\Gamma)$.
\end{enumerate}
 %----------------%
In particular, $\lambda_1 (\Gamma) \leq \lambda_1 (\widetilde \Gamma)$.
\end{proposition}
 %----------------%
\begin{proof}
(i) Assume that $\lambda_k (\widetilde \Gamma) \leq 0$ and pick a subspace $\widetilde \cF$ of $\dom \widetilde \sa$ of dimension $k$ such that
 %----------------%
\begin{align*}
 \widetilde \sa [\widetilde f] \leq \lambda_k (\widetilde \Gamma) \|\widetilde f\|_{L^2 (\widetilde \Gamma)}^2
\end{align*}
 %----------------%
holds for all $\widetilde f \in \widetilde \cF$. Let $\cF \subset \dom \sa$ denote the space of restrictions to $\Gamma$ of functions in $\widetilde \cD$. The restriction $f$ of any non-zero $\widetilde f \in \widetilde \cF$ to $\Gamma$ is obviously non-zero again, hence $\dim \cF = \dim \widetilde \cF = k$. Moreover, we have
 %----------------%
\begin{align*}
 \sa [f] \leq \sa [f] + \int_{\widetilde \Gamma \setminus \Gamma} |\widetilde f'|^2 {\,\dd x} = \widetilde \sa [\widetilde f] \leq \lambda_k (\widetilde \Gamma) \|\widetilde f\|_{L^2 (\widetilde \Gamma)}^2 \leq \lambda_k (\widetilde \Gamma) \|f\|_{L^2 (\Gamma)}^2,
\end{align*}
 %----------------%
where the last estimate comes from $\|f\|_{L^2 (\Gamma)} \leq \|\widetilde f\|_{L^2 (\widetilde \Gamma)}$ and $\lambda_k (\widetilde \Gamma) \leq 0$. Hence using the Rayleigh--Ritz principle we infer that
 %----------------%
\begin{align*}
 \lambda_k (\Gamma) \leq \max_{\substack{f \in \cF\\ f \neq 0}} \frac{\sa [f]}{\|f\|_{L^2 (\Gamma)}^2} \leq \lambda_k (\widetilde \Gamma).
\end{align*}
 %----------------%

(ii) Assume next that $\lambda_k (\Gamma) \geq 0$ and take a subspace $\cF$ of $\dom \sa$ with $\dim \cF = k$ such that
 %----------------%
\begin{align*}
 \sa [f] \leq \lambda_k (\Gamma) \|f\|_{L^2 (\Gamma)}^2
\end{align*}
 %----------------%
holds for all $f \in \cF$. For each $f \in \cF$, denote by $\widetilde f$ the extension of $f$ to $\widetilde \Gamma$ by a constant on the added part, continuous at the gluing vertex, and let $\widetilde \cF$ be the $k$-dimensional subspace of $\dom \widetilde \sa$ consisting of all such extensions of all $f \in \cF$. Then
 %----------------%
\begin{align*}
 \widetilde \sa [\widetilde f] & = \sa [f] \leq \lambda_k (\Gamma) \|f\|_{L^2 (\Gamma)}^2 \leq \lambda_k (\Gamma) \|\widetilde f\|_{L^2 (\widetilde \Gamma)}^2
\end{align*}
 %----------------%
holds for all $\widetilde f \in \widetilde \cF$, implying the assertion by the same argument as above.
\end{proof}

%%%%%%%%%%%%%%%%%%%%%%%%%%%%%%%%%%%%%%%%%%%%%%%%%%%
\section{Remarks on star graphs with Dirichlet endpoints}
\label{s:Dirichlet}

So far we have considered star graphs with Neumann conditions at the `loose ends'. Star graphs with Dirichlet endpoints behave differently, and we will discuss this briefly in this section. To be precise, we consider, on a finite metric star graph $\Gamma$, the negative Laplacian with vertex conditions \eqref{eq:conditions} at the central vertex $v_0$ and Dirichlet conditions at the vertices $v_1, \dots, v_N$ of degree one. It is easy to see that this operator corresponds to the sesquilinear form
\begin{align*}
 \sa_{\rm D} [f, g] & = \sum_{j = 1}^N \int_0^{l_j} f_j' \overline{g_j'}{\,\dd x} + i \sum_{j = 2}^N \sum_{k = 1}^{j - 1} (-1)^{j + k} \left( F_k \overline{G_j} - F_j \overline{G_k} \right)
\end{align*}
 %----------------%
with domain
 %----------------%
\begin{align*}
 \dom \sa_{\rm D} = \begin{cases}
             \big\{ f \in \widetilde H^1 (\Gamma) : f (v_j) = 0, j = 1, \dots, N \big\}, & \text{if}~N~\text{ odd},\\
             \Big\{ f \in \widetilde H^1 (\Gamma) : \sum_{j = 1}^N (-1)^j F_j = 0, f (v_j) = 0, j = 1, \dots, N \Big\},& \text{if}~N~\text{even}.
            \end{cases}
\end{align*}

The following proposition should be compared to Proposition \ref{p:extens}.
 %----------------%
\begin{proposition}
Consider the Laplacian with vertex conditions \eqref{eq:conditions} at the central vertex and Dirichlet boundary conditions at the vertices of degree one. If $\widetilde \Gamma$ is obtained from the star graph $\Gamma$ by extending the length of an edge, the eigenvalues satisfy the inequality
\begin{align*}
 \lambda_k (\widetilde \Gamma) \leq \lambda_k (\Gamma) \;\;\, \text {for all}\;\; k\in\dN.
\end{align*}
 %----------------%
\end{proposition}
 %----------------%
\begin{proof}
For any $f \in \dom \sa_{\rm D}$, the extension of $f$ by zero to $\widetilde \Gamma$, $\widetilde f$, belongs to $\dom \widetilde \sa_{\rm D}$  (the form on $\widetilde \Gamma$), and
 %----------------%
\begin{align*}
 \widetilde \sa_{\rm D} [\widetilde f] = \sa_{\rm D} [f], \quad \|\widetilde f\|_{L^2 (\widetilde \Gamma)} = \|f\|_{L^2 (\Gamma)}.
\end{align*}
 %----------------%
Take now $\cF \subset \dom \sa$ a $k$-dimensional subspace for which
 %----------------%
\begin{align*}
 \sa_{\rm D} [f] \leq \lambda_k (\Gamma) \|f\|_{L^2 (\Gamma)}^2, \quad f \in \cF.
\end{align*}
 %----------------%
Furthermore, let $\widetilde \cF$ be the $k$-dimensional subspace of $\dom \widetilde \sa$ obtained from zero-extensions of functions in $\cF$. Then we have
 %----------------%
\begin{align*}
 \widetilde \sa_{\rm D} [\widetilde f] = \sa_{\rm D} [f] \leq \lambda_k (\Gamma)\|f\|_{L^2 (\Gamma)}^2 = \lambda_k (\Gamma) \|\widetilde f\|_{L^2 (\widetilde \Gamma)}^2, \quad \widetilde f \in \widetilde \cF,
\end{align*}
 %----------------%
which implies the claim in the same way as in Proposition~\ref{p:extens}.
\end{proof}

With Proposition~\ref{p:neg_exist} in mind, it is natural to ask whether the star with Dirichlet endpoints has negative eigenvalues. As we will see in the following example, this is not always the case.
 %----------------%
\begin{example} \label{ex: dstar}
In the Dirichlet case, the ansatz
 %----------------%
\begin{align*}
 \psi_j (x) = \alpha_j \sinh (\kappa (l_j - x))
\end{align*}
 %----------------%
can be used. Then
 %----------------%
\begin{align*}
 \psi_j' (x) = - \kappa \alpha_j \cosh(\kappa (l_j - x))
\end{align*}
 %----------------%
and the vertex conditions at the central vertex give
 %----------------%
\begin{align*}
 \alpha_{j + 1} \sinh (\kappa l_{j+1}) - \alpha_j \sinh (\kappa l_j) - i \kappa \big( \alpha_j \cosh (\kappa l_j) + \alpha_{j+1} \cosh (\kappa l_{j+1}) \big) = 0
\end{align*}
 %----------------%
for $j = 1, \dots, N$. For the equilateral case $l_1 = \dots = l_N = l$, let $A = \sinh (\kappa l) + i \kappa \cosh (\kappa l)$ and $B = - \sinh (\kappa l) + i \kappa \cosh (kl)$. Then the coefficients $\alpha_j$ satisfy
 %----------------%
\begin{align*}
 \begin{pmatrix} A & B & 0 & \hdots & \hdots & 0 \\ 0 & \ddots & \ddots & \ddots &  & \vdots \\ \vdots & \ddots & \ddots & \ddots & \ddots & \vdots \\ \vdots & & \ddots & \ddots & \ddots & 0 \\ 0 & \hdots & \hdots & 0 & A & B \\ B & 0 & \hdots & \hdots & 0 & A \end{pmatrix} \begin{pmatrix} \alpha_1 \\ \vdots \\ \alpha_N \end{pmatrix} = 0.
\end{align*}
 %----------------%
By computing the determinant $A^N + (-1)^{N + 1} B^N$  of the involved matrix, we obtain the secular equation
 %----------------%
\begin{align*}
 \left( \sinh (\kappa l) + i \kappa \cosh (\kappa l) \right)^N + (-1)^{N+1} \left(-  \sinh (\kappa l) + i \kappa \cosh (kl) \right)^N = 0.
\end{align*}
 %----------------%
For $N = 3$ we get
 %----------------%
\begin{align*}
 \kappa \left(3 \sinh^2 (\kappa l) \cosh (\kappa l) - \kappa^2 \cosh^3 (\kappa l) \right) = 0
\end{align*}
 %----------------%
or equivalently
 %----------------%
\begin{align*}
 3 \tanh^2 (\kappa l) - \kappa^2 = 0.
\end{align*}
 %----------------%
For $\kappa > 0$ this is satisfied if and only if $\sqrt{3} \tanh (\kappa l) = \kappa$. Since the left-hand side of this equation has derivative $\sqrt{3} l$ (w.r.t.\ $\kappa$) at $\kappa = 0$ and is concave, a solution exists if and only if $l >3^{-1/2}$, and it is unique in this case. That is, a unique negative eigenvalue exists if and only if $l>3^{-1/2}$. For $l = 1$, e.g., we get $\kappa \approx 1.60$, and the larger $l$ the smaller is the first eigenvalue in accordance with the previous proposition; in the limit $l \to \infty$ it converges from above to the eigenvalue $-3$ of the infinite 3-star graph corresponding to $\kappa=\sqrt{3}$. If the edges are short, on the other hand, meaning $l<3^{-1/2}$, the spectrum is positive; replacing $\kappa$ by $-ik$ with $k\in\dR$, we get the secular equation
 %----------------%
\begin{align*}
 3 \tan^2 (kl) - k^2 = 0.
\end{align*}
 %----------------%
Furthermore, for $N = 3$ not necessarily equilateral, we get the secular equation
 %----------------%
\begin{align*}
 \tanh (\kappa l_1) \tanh (\kappa l_2) + \tanh (\kappa l_1) \tanh (\kappa l_3) +  \tanh (\kappa l_2) \tanh (\kappa l_3) = \kappa^2,
\end{align*}
 %----------------%
and its counterpart for positive eigenvalues.
% Here numerics seems to indicate that for fixed total length the equilateral 3-star may have the smallest ground state eigenvalue. E.g., lengths $1.5, 0.8, 0.7$ give $\kappa \approx 1.49$. \marginpar{\tiny This needs more attention even if it is probably true. Evaluating the Hessian one finds that the left-hand side is concave for large enough values of the $l_j$s}
\end{example}

\section{Eigenvalue bounds for general compact metric graphs}
\label{sec:generalGraphs}

Now we turn to more general (not necessarily star) graphs. In this section, $\Gamma$ is a compact, finite metric graph with edge set $\cE$ of cardinality $E$ and vertex set $\cV$. We interpret $\Gamma$ as a finite collection of intervals $[0, l (e)]$, $e \in \cE$, meeting at certain of their endpoints identified with graph vertices. This view point will be natural, for instance, in the context of the graph surgery operations mentioned already in Section~\ref{s: surgery}; especially we will look into an operation where vertices of a given graph will be ``glued together'' to form a new graph. We stress that the vertex conditions and the operators we will consider are independent of the choice of the edges' orientations.

For every edge $e \mathrel{\widehat{=}} [0, l (e)]$ we denote by $o (e)$ the vertex associated with its zero endpoint (``the vertex from which $e$ originates'') and by $t (e)$ the vertex associated with the endpoint $l (e)$ (``the vertex at which $e$ terminates''); note that $o (e) = t (e)$ if $e$ is a loop. For every vertex $v \in \cV$ we write
 %----------------%
\begin{align*}
 \cE_{v, o} := \left\{e \in \cE : o (e) = v \right\} \quad \text{and} \quad \cE_{v, t} := \left\{e \in \cE : t (e) = v \right\}.
\end{align*}
 %----------------%
% Occasionally it will be useful to have a notation for the respective endpoints of edges at hand; we
% Let us choose an enumeration $\{v_1, \dots, v_n\}$ of the vertices of $\Gamma$.
We will again say that a function $f : \Gamma \to \C$ belongs to $\widetilde H^k (\Gamma)$ for some $k \geq 0$ if its restriction to any edge $e \in \cE$ belongs to the Sobolev space $H^k (0, l (e))$. Let $f \in \widetilde H^1 (\Gamma)$. In order to define vertex conditions, for every vertex $v \in \cV$ we fix an enumeration of the endpoints meeting at $v$ and define, according to this enumeration, a vector $F (v) := (F_1 (v), \dots, F_{\deg (v)} (v))^\top$ of boundary values of $f$ at the vertex $v$. Note that while sometimes there may be a ``natural'' enumeration as in the case of a planar graph, in general we do not assume $\Gamma$ to be embedded in a Euclidean space and the enumeration choice is arbitrary. Analogously, if $f \in \widetilde H^2 (\Gamma)$, then we can define a vector $F' (v) = (F_1' (v), \dots, F_{\deg (v)}' (v))^\top$ of derivatives of $f$ at the endpoints corresponding to $v$ (conventionally, each edge is identified with an interval and the derivative is taken at its left endpoint).

At vertices of degree $\deg (v) \geq 2$ we again consider the matching conditions
 %----------------%
\begin{align}\label{eq:conditionsGeneral}
 \left( F_{j + 1} (v) - F_j (v) \right) + i \left( F_j' (v) + F_{j+1}' (v) \right) = 0, \quad j = 1, \dots, \deg (v).
\end{align}
 %----------------%
Furthermore, at any vertex of degree one we impose Neumann conditions, $f' (v) = 0$. We point out once more that the conditions \eqref{eq:conditionsGeneral} in the case of a vertex of degree one formally reduce to Neumann conditions. We will therefore in this section take the view point of having conditions \eqref{eq:conditionsGeneral} at all vertices, which will simplify some considerations.

 %----------------%
\begin{example}
Let us derive a secular equation for the negative eigenvalues of the operator under consideration. To do so, we define for a given $\lambda < 0$ a Dirichlet-to-Neumann matrix $M (\lambda) \in \C^{2 E \times 2 E}$ via
 %----------------%
\begin{align}\label{eq:M}
 M (\lambda) \begin{pmatrix} F (v_1) \\ \vdots \\ F (v_n) \end{pmatrix} = \begin{pmatrix} F' (v_1) \\ \vdots \\ F' (v_n) \end{pmatrix},
\end{align}
 %----------------%
where $v_1, \dots, v_n$ is an enumeration of the vertices of $\Gamma$ and $f \in \widetilde H^2 (\Gamma)$ satisfies $- f'' = \lambda f$ on every edge; the fact that $\lambda$ is negative, and consequently, not a Dirichlet Laplacian eigenvalue for the disconnected edges, implies that $M (\lambda)$ is well-defined. The matrix $M (\lambda)$ can be computed as follows: the chosen enumeration of endpoints associated with each of the vertices $v_1, v_2, \dots, v_n$ gives rise to an enumeration of all the edges' endpoints. According to this enumeration, for $i < j$ the $2 \times 2$ submatrix
 %----------------%
\begin{align*}
 \begin{pmatrix} m_{ii} (\lambda) & m_{ij} (\lambda) \\ m_{ji} (\lambda) & m_{jj} (\lambda) \end{pmatrix}
\end{align*}
 %----------------%
is either zero, if the $i$-th and $j$-th endpoint belong to different edges, or equals the $2 \times 2$ Dirichlet-to-Neumann matrix for the interval $(0, l (e))$, if the $i$-th and $j$-th endpoint belong to the same edge $e$.\footnote{The Dirichlet-to-Neumann matrix for an interval $I = (0, l)$ can be computed explicitly and is given by $M_I (\lambda) = \kappa \begin{pmatrix} - \coth (\kappa l) & \frac{1}{\sinh (\kappa l)} \\ \frac{1}{\sinh (\kappa l)} & - \coth (\kappa l)  \end{pmatrix}$ where $\lambda = - \kappa^2$.} In particular, each row and column of $M (\lambda)$ has precisely two non-zero entries.

For $m = 1, \dots, n$ we now define the matrices
 %----------------%
\begin{align*}
 A_m = \begin{pmatrix} 1 & -1 & 0 & \cdots & \cdots & 0 \\
 0 & \ddots & \ddots & \ddots &  & \vdots \\
 \vdots & \ddots & \ddots & \ddots & \ddots & \vdots \\
 \vdots &  & \ddots & \ddots & \ddots & 0 \\
 0 & \cdots & \cdots & 0 & \ddots & -1 \\
 -1 & 0 & \cdots & \cdots & 0 & 1
 \end{pmatrix}
 \in \C^{\deg (v_m) \times \deg (v_m)}
\end{align*}
 %----------------%
in case $\deg (v_m) \geq 2$, or $A_m = (0)$ in case $\deg (v_m) = 1$, as well as
 %----------------%
\begin{align*}
 B_m = \begin{pmatrix} 1 & 1 & 0 & \cdots & \cdots & 0 \\
 0 & \ddots & \ddots & \ddots &  & \vdots \\
 \vdots & \ddots & \ddots & \ddots & \ddots & \vdots \\
 \vdots &  & \ddots & \ddots & \ddots & 0 \\
 0 & \cdots & \cdots & 0 & \ddots & 1 \\
 1 & 0 & \cdots & \cdots & 0 & 1
 \end{pmatrix}
 \in \C^{\deg (v_m) \times \deg (v_m)}.
\end{align*}
 %----------------%
if $\deg (v_m) \geq 2$ and $B_m = (1)$ in case $\deg (v_m) = 1$. Moreover, let
 %----------------%
\begin{align*}
 \cA := \diag (A_1, \dots, A_n) \quad \text{and} \quad \cB := \diag (B_1, \dots, B_n)
\end{align*}
 %----------------%
be the block diagonal matrices composed of the $A_m$ and $B_m$, respectively. Note that the vertex conditions at $v_m$ can be written equivalently as
 %----------------%
\begin{align*}
 A_m F (v_m) + i B_m F' (v_m) = 0.
\end{align*}
 %----------------%
This means that a function $f \in \widetilde H^2 (\Gamma)$ satisfies the vertex conditions under consideration at each vertex if and only if
 %----------------%
\begin{align*}
 \cA \begin{pmatrix} F (v_1) \\ \vdots \\ F (v_n) \end{pmatrix} + i \cB \begin{pmatrix} F' (v_1) \\ \vdots \\ F' (v_n) \end{pmatrix} = 0
\end{align*}
 %----------------%
holds. Note that due to the block structure of the matrices $\cA$ and $\cB$ the chosen order of vertices is irrelevant.

We can now easily derive the sought secular equation. Indeed, for a given $\lambda < 0$, a function $f \in \widetilde H^2 (\Gamma)$ solving $- f'' = \lambda f$ on each edge satisfies the indicated vertex conditions in view of \eqref{eq:M} if and only if
 %----------------%
\begin{align*}
 \left( \cA + i \cB M (\lambda) \right) \begin{pmatrix} F (v_1) \\ \vdots \\ F (v_n) \end{pmatrix} = 0.
\end{align*}
 %----------------%
Having in mind that $\lambda$ is not a Dirichlet Laplacian eigenvalue of any edge, it follows that $\lambda < 0$ is an eigenvalue of the operator under consideration if and only if
 %----------------%
\begin{align}\label{eq:secular}
 \det (\cA + i \cB M (\lambda)) = 0.
\end{align}
 %----------------%
\end{example}
 %----------------%
\begin{remark}\label{rem:order}
The vertex conditions \eqref{eq:conditionsGeneral} depend on a fixed choice of an order of the endpoints of edges incident in $v$. However, the spectrum of the corresponding operator is independent of it. To explain it, let $\lambda$ be a negative eigenvalue satisfying the secular equation \eqref{eq:secular}. For a given vertex $v_m$, interchanging the chosen order of two incident edges results in interchanging the corresponding columns in the matrices $A_m$ and $B_m$, as well as interchanging both the corresponding rows and columns in $M (\lambda)$, cf.\ the expression of $M (\lambda)$ in the previous example. As a result, interchanging the position of two endpoints corresponding to the same vertex leads to interchange of the corresponding columns in the matrix $\cA + i \cB M (\lambda)$. Since this, up to a sign, does not change its determinant, it follows from \eqref{eq:secular} that the negative spectrum is invariant under a change of enumeration of endpoints within a vertex. The same is true for positive eigenvalues; as in Example~\ref{ex: dstar} we can replace $\kappa$ by $-ik$ with $k\in\dR$. Of course, in contrast to negative energies there are now points where $M (\lambda)$ is not well defined, referring Dirichlet eigenvalues on edges. However, they constitute a discrete set to which one can extend the claim using a continuity argument; recall that the eigenvalues of the graph Laplacian with a fixed vertex coupling are analytic with respect to the edge lengths \cite[Sec.~2.5.1]{BK13}.
\end{remark}
 %----------------%
 The following example will be relevant, as an object of comparison, in the proof of Theorem \ref{thm:optimizationGeneral} below.

\begin{example}\label{ex:figure8}
Let us consider a figure-8 graph, i.e.\ a graph with one vertex and two loops meeting at it. Call the two edge lengths $l_1$ and $l_2$. In order to compute a potential negative eigenvalue, we use the ansatz
 %----------------%
\begin{align*}
 \psi_j (x) = \alpha_j \cosh (\kappa x) + \beta_j \sinh (\kappa x), \quad x \in [0, l_j],
\end{align*}
 %----------------%
so that
 %----------------%
\begin{align*}
 \psi_j' (x) = \kappa \left( \alpha_j \sinh (\kappa x) + \beta_j \cosh (\kappa x) \right), \quad x \in [0, l_j],
\end{align*}
 %----------------%
holds for $j = 1, 2$. A straightforward computation using the conditions \eqref{eq:conditionsGeneral} (where the chosen order starts with the two endpoints of $e_1$ followed by those of $e_2$) leads to the equations
 %----------------%
\begin{align*}
 \begin{pmatrix}
  A_1 & B_1 & 0 & 0 \\
  A_1' & B_1' & 1 & i \kappa \\
  0 & 0 & A_2 & B_2  \\
  1 & i \kappa & A_2' & B_2'
 \end{pmatrix}
 \begin{pmatrix}
  \alpha_1 \\ \beta_1 \\ \alpha_2 \\ \beta_2
 \end{pmatrix}
 = 0,
\end{align*}
 %----------------%
with the matrix entries
 %----------------%
\begin{align*}
 A_j & = \cosh(\kappa l_j) - 1 - i \kappa \sinh(\kappa l_j), \\
 B_j & = \sinh (\kappa l_j) + i \kappa (1 - \cosh (\kappa l_j)), \\
 A_j' & = - \cosh (\kappa l_j) - i \kappa \sinh (\kappa l_j), \\
 B_j' & = - \sinh (\kappa l_j) - i \kappa \cosh (\kappa l_j)
\end{align*}
 %----------------%
for $j = 1, 2$. Taking the determinant of the system's matrix, we obtain that $- \kappa^2$ is an eigenvalue if and only if
 %----------------%
\begin{align*}
 -16 i \kappa (\kappa^2 - 1) \sinh \frac{\kappa l_1}{2} \sinh \frac{\kappa l_2}{2} \sinh \frac{\kappa (l_1 + l_2)}{2} = 0.
\end{align*}
 %----------------%
Consequently, $\lambda_1 (\Gamma) = - 1$ is the only negative eigenvalue, independently of the choice of the edge lengths $l_1$ and $l_2$ and the chosen enumeration of the edges' endpoints at the vertex; note that it coincides with the negative eigenvalue of an infinite 4-star \cite{ET18}.

Next we observe that the figure-8 graph has an eigenvalue zero with multiplicity one, the eigenspace containing all constant functions.

For the positive eigenvalues $\lambda = k^2$ with $k > 0$, we make an ansatz
\begin{align*}
 \psi_j (x) = \alpha_j \cos (k x) + \beta_j \sin (k x).
\end{align*}
A computation similar to the one above shows that the vertex conditions \eqref{eq:conditionsGeneral} lead to the system of equations
\begin{align*}
 \begin{pmatrix}
  A_1 & B_1 & 0 & 0 \\
  A_1' & B_1' & 1 & i k \\
  0 & 0 & A_2 & B_2 \\
  1 & i k & A_2' & B_2'
 \end{pmatrix}
 \begin{pmatrix}
  \alpha_1 \\ \beta_1 \\ \alpha_2 \\ \beta_2
 \end{pmatrix}
 = 0,
\end{align*}
where
\begin{align*}
 A_j & = \cos (k l_j) - 1 + i k \sin (k l_j), \\
 B_j & = \sin (k l_j) + i k \left( 1 - \cos (k l_j) \right), \\
 A_j' & = - \cos (k l_j) + i k \sin (k l_j), \\
 B_j' & = - \sin (k l_j) - i k \cos (k l_j).
\end{align*}
Taking the determinant we find that $\lambda = k^2 > 0$ is an eigenvalue if and only if
\begin{align*}
 - 4 i k (1 + k^2) \big( \sin (k l_1) + \sin (k l_2) - \sin (k (l_1 + l_2)) \big) = 0,
\end{align*}
that is, essentially,
\begin{align*}
 \sin (k l_1) + \sin (k l_2) - \sin (k (l_1 + l_2)) = 0.
\end{align*}
From this we see immediately that the positive eigenvalues enjoy a scaling property: multiplying the lengths of both edges by a factor $t > 0$ leads to a scaling of the positive eigenvalues by the factor $t^{-2}$.

We can use this to compute some of the positive eigenvalues explicitly. Due to the scaling property we may assume without loss of generality that $l_1 + l_2 = 1$. In this case the secular equation is
\begin{align}\label{eq:secularSpecial}
 \sin (k) = \sin (k l_1) + \sin (k (1 - l_1)),
\end{align}
and this is clearly satisfied if $k = 2 n \pi$ for some $n \in \N$, since in this case both sides vanish due to the periodicity and anti-symmetry of $\sin$. There may exist further eigenvalues in general, which can easily be checked for specific choices of $l_1$. However, in the equilateral case $l_1 = l_2 = \frac{1}{2}$ the equation \eqref{eq:secularSpecial} simplifies $\sin (k) = 2 \sin (\frac{k}{2})$, and the positive eigenvalues are given by the numbers $(2 n \pi)^2$, $n \in \N$, with multiplicity one if $n$ is odd and multiplicity three if $n$ is even.
\end{example}
 %----------------%

Next we extend Proposition~\ref{prop:form} and express the sesquilinear form associated with the operator under consideration. We omit its proof which is analogous to that of Proposition \ref{prop:form}.
 %----------------%
\begin{proposition}\label{prop:formGeneral}
Let $v_1, \dots, v_n$ be an enumeration of $\cV$. Then the sesquilinear form associated with the Laplacian on $\Gamma$ with vertex conditions \eqref{eq:conditionsGeneral} at each vertex of degree larger or equal two and Neumann boundary conditions at vertices of degree one is given by
\begin{align*}
 \sa [f, g] & = \sum_{j = 1}^{|\cE|} \int_0^{l_j} f_j' \overline{g_j'} {\,\dd x} \\
 & \quad + i \sum_{m = 1}^n \sum_{j = 2}^{\deg (v_m)} \sum_{k = 1}^{j - 1} (-1)^{j + k} \left( F_k (v_m) \overline{G_j (v_m)} - F_j (v_m) \overline{G_k (v_m)} \right).
\end{align*}
Its domain consists of all $f \in \widetilde H^1 (\Gamma)$ such that
\begin{align*}
 \sum_{j = 1}^{\deg (v_m)} (-1)^j F_j (v_m) = 0
\end{align*}
holds for all $m$ such that $\deg (v_m)$ is even.
\end{proposition}

Our aim is now to provide upper bounds for the eigenvalues $\lambda_1 (\Gamma) \leq \lambda_2 (\Gamma) \leq \dots$ of the negative Laplacian on any finite, compact metric graph $\Gamma$ with vertex conditions \eqref{eq:conditionsGeneral}. In order to do so, we establish some surgery principles.
 %----------------%
\begin{definition}\label{def:join}
Let $\Gamma$ be a finite, compact metric graph with vertex set $\cV (\Gamma)$ and let $v_1, v_2 \in \cV (\Gamma)$. We say that {\it $\widetilde \Gamma$ is obtained from $\Gamma$ by merging $v_1$ and $v_2$} or, equivalently, that {\it $\Gamma$ can be obtained from $\widetilde \Gamma$ by splitting a vertex $\hat v$} if the following hold: $\widetilde \Gamma$ has the same set of edges as $\Gamma$, $\cE (\widetilde \Gamma) = \cE (\Gamma)$; its vertex set equals
 %----------------%
\begin{align*}
 \cV (\widetilde \Gamma) = (\cV (\Gamma) \setminus \{v_1, v_2\}) \cup \{\hat v\},
\end{align*}
 %----------------%
where $\hat v$ is an additionally introduced vertex, and the relations between edges and vertices in $\widetilde \Gamma$ are given by
 %----------------%
\begin{align*}
 \cE_{v, o} (\widetilde \Gamma) = \cE_{v, o} (\Gamma) \quad \text{and} \quad \cE_{v, t} (\widetilde \Gamma) = \cE_{v, t} (\Gamma), \quad v \neq \hat v,
\end{align*}
 %----------------%
as well as
 %----------------%
\begin{align*}
 \cE_{\hat v, o} (\widetilde \Gamma) = \cE_{v_1, o} (\Gamma) \cup \cE_{v_2, o} (\Gamma) \quad \text{and} \quad \cE_{\hat v, t} (\widetilde \Gamma) = \cE_{v_1, t} (\Gamma) \cup \cE_{v_2, t} (\Gamma).
\end{align*}
 %----------------%
\end{definition}

Using this definition, we prove the following surgery principles on the effect of merging vertices on the spectrum:
 %----------------%
\begin{proposition}\label{prop:join}
Let $\widetilde \Gamma$ be obtained from $\Gamma$ by merging two vertices $v_1, v_2 \in \cV$. Then the following claims hold:
 %----------------%
\begin{enumerate}
 \item If $\deg (v_1) + \deg (v_2)$ is odd, i.e.\ $v_1$ and $v_2$ have opposite parity, then
 %----------------%
 \begin{align}\label{eq:joining}
  \lambda_k (\widetilde \Gamma) \leq \lambda_k (\Gamma), \quad k \in \N,
 \end{align}
 %----------------%
 holds.
 \item If both $\deg (v_1)$ and $\deg (v_2)$ are even, then \eqref{eq:joining} holds again.
 \item If both $\deg (v_1)$ and $\deg (v_2)$ are odd, then we have
  %----------------%
 \begin{align}\label{eq:joining2}
  \lambda_k (\Gamma) \leq \lambda_k (\widetilde \Gamma), \quad k \in \N.
 \end{align}
 %----------------%
\end{enumerate}
 %----------------%
\end{proposition}
 %----------------%
\begin{proof}
We denote by $\sa$ and $\widetilde \sa$ the quadratic forms of the graphs $\Gamma$ and $\widetilde \Gamma$, respectively, according to Proposition \ref{prop:formGeneral}.

(i) Take $\psi \in \dom \sa$. As we have $\cE (\widetilde \Gamma) = \cE (\Gamma)$, $\psi$ can be also understood as a function defined on $\widetilde \Gamma$. Denote by $\hat v$ the vertex in $\widetilde \Gamma$ which has resulted from merging $v_1$ and $v_2$ as in Definition \ref{def:join}. Since $\deg (v) = \deg (v_1) + \deg (v_2)$ is odd by assumption, we have obviously $\psi \in \dom \widetilde \sa$. As pointed out in Remark \ref{rem:order}, the order of edges incident in $\hat v$ where we impose the vertex conditions \eqref{eq:conditionsGeneral} does not influence the eigenvalues; therefore we may choose the enumeration in which the indices $1, \dots, d_1$ correspond to the vertex $v_1$ in $\Gamma$ and the indices $d_1 + 1, \dots, d_1 + d_2$ correspond to $v_2$, where $d_j = \deg (v_j)$, $j = 1, 2$. Moreover, without loss of generality we may suppose that $d_1$ is even and $d_2$ is odd. Then we have
 %----------------%
\begin{align*}
 \widetilde \sa [\psi] - \sa [\psi] & = 2 \left( \sum_{j = 2}^{d_1 + d_2} \sum_{k = 1}^{j - 1} - \sum_{j = 2}^{d_1} \sum_{k = 1}^{j - 1} - \sum_{j = d_1 + 2}^{d_1 + d_2} \sum_{k = d_1 + 1}^{j - 1} \right) (-1)^{j + k - 1} \Imag (\Psi_k (\hat v) \overline{\Psi_j} (\hat v)) \\
 & = 2 \sum_{j = d_1 + 1}^{d_1 + d_2} \sum_{k = 1}^{d_1} (-1)^{j + k - 1} \Imag (\Psi_k (\hat v) \overline{\Psi_j} (\hat v)) \\
 & = 2 \Imag \sum_{j = d_1 + 1}^{d_1 + d_2} \overline{\Psi_j} (\hat v) \sum_{k = 1}^{d_1} (-1)^{j + k - 1} \Psi_k (\hat v) = 0,
\end{align*}
 %----------------%
as $\psi \in \dom \sa$ and $d_1 = \deg (v_1)$ is supposed to be even. Thus $\dom \sa \subset \dom \widetilde \sa$ and $\widetilde \sa [\psi] = \sa [\psi]$ holds for all $\psi \in \dom \sa$; this yields inequality \eqref{eq:joining}.

(ii) Using the same notation as in part (i), as both $d_1$ and $d_2$ are even each vector $\psi \in \dom \sa$ satisfies
 %----------------%
\begin{align*}
 \sum_{j = 1}^{d_1} (-1)^j \Psi_j (v_1) = 0 = \sum_{j = d_1 + 1}^{d_2} (-1)^j \Psi_j (v_2),
\end{align*}
 %----------------%
and consequently
 %----------------%
\begin{align*}
 \sum_{j = 1}^{d_1 + d_2} (-1)^j \Psi_j (\hat v) = 0;
\end{align*}
 %----------------%
thus $\dom \sa \subset \dom \widetilde \sa$ holds. The same computation as in the proof of (i) shows that the two forms coincide on the smaller domain, hence \eqref{eq:joining} follows.

(iii) Assume finally that both $d_1$ and $d_2$ are odd, and let $\psi \in \dom \widetilde \sa$. Then $\psi \in \dom \sa$ as no restrictions are imposed at $v_1$ and $v_2$, and furthermore, the same computation as in (i), together with the fact that
 %----------------%
\begin{align*}
 \sum_{k = 1}^{d_1 + d_2} (-1)^k \Psi_k (\hat v) = 0
\end{align*}
 %----------------%
holds, gives
 %----------------%
\begin{align*}
 \widetilde \sa [\psi] - \sa [\psi] & = 2 \Imag \sum_{j = d_1 + 1}^{d_1 + d_2} \overline{\Psi_j} (v_2) \sum_{k = 1}^{d_1} (-1)^{j + k - 1} \Psi_k (v_1) \\
 & = 2 \Imag  \sum_{j = d_1 + 1}^{d_1 + d_2} \overline{\Psi_j} (v_2) \sum_{k = d_1 + 1}^{d_1 + d_2} (- 1)^{j + k} \Psi_k (v_2) \\
 & = 2 \Imag \bigg| \sum_{j = d_1 + 1}^{d_1 + d_2} (- 1)^j \Psi_j (v_2) \bigg|^2 = 0.
\end{align*}
 %----------------%
This means that $\dom \widetilde \sa \subset \dom \sa$ and the two forms coincide on the smaller domain; in this way we have shown \eqref{eq:joining2} concluding thus the proof.
\end{proof}

We are now in position to prove upper bounds for all eigenvalues. 
% \jr{\sout{To simplify the formulation of the following theorem, note that without loss of generality we may suppose that $\Gamma$ has no vertices of degree two, since for them the vertex conditions \eqref{eq:conditionsGeneral} reduce to continuity-Kirchhoff conditions on a line which can trivially be neglected.}}
 %----------------%
\begin{theorem}\label{thm:optimizationGeneral}
Let $\Gamma$ be a finite, compact metric graph  different from a path graph or a cycle graph. Assume that at each vertex of degree two or larger vertex conditions \eqref{eq:conditionsGeneral} are imposed while Neumann boundary conditions are chosen at the vertices of degree one. Then the inequalities
 %----------------%
\begin{align*}
 \lambda_1 (\Gamma) \leq - 1, \qquad \lambda_2 (\Gamma) \leq 0 \qquad \text{and} \qquad \lambda_{2 k + 1} (\Gamma) \leq \frac{4 k^2 \pi^2}{L (\Gamma)^2}, \quad k = 1, 2, \dots\,,
\end{align*}
 %----------------%
hold. In particular, there exists at least one negative eigenvalue.
\end{theorem}
 %----------------%
\begin{proof}
The proof is carried out by performing surgery operations, more precisely merging and splitting vertices appropriately. First we successively join pairs of vertices of odd degree, an operation under which all eigenvalues behave non-decreasingly, see Proposition \ref{prop:join} (iii). This can be repeated until we arrive at a graph $\Gamma'$ in which every vertex has even degree; note that since the number of edge endpoints is even, it cannot happen we are left with $\Gamma'$ having a single vertex of odd degree. Also, since $\Gamma$ is neither a path graph nor a cycle, it has a vertex of degree three or higher, implying that $\Gamma'$ contains at least one vertex $v'$ of degree four or higher. Furthermore, $\Gamma'$ contains an Eulerian cycle, since each vertex of $\Gamma'$ has even degree. We use this and split all vertices except $v'$ successively into even degree vertices, in a way such that the Eulerian cycle remains closed; we continue until every vertex except $v'$ has degree two. The vertex $v'$ we split successively into a number of vertices of degree two and one vertex of degree four, splitting in such a way that the Eulerian cycle remains connected. Under these splitting operations, again all eigenvalues behave non-decreasingly by Proposition \ref{prop:join} (ii). The resulting graph $\Gamma''$ contains one vertex of degree four and only vertices of degree two otherwise, in other words, it can be identified with a figure-8 graph. In combination with Example \ref{ex:figure8} we obtain
 %----------------%
\begin{align*}
 \lambda_1 (\Gamma) \leq \lambda_1 (\Gamma'') = - 1 \quad \text{and} \quad \lambda_2 (\Gamma) \leq \lambda_2 (\Gamma'') = 0.
\end{align*}
 %----------------%
Finally, if we split the only vertex of $\Gamma''$ once more into two vertices of degree two, we effectively obtain a cycle graph $\Gamma'''$ with two vertices and continuity-Kirchhoff conditions at both of them. Hence another application of Proposition \ref{prop:join}, again claim (ii), yields $\lambda_j (\Gamma) \leq \lambda_j (\Gamma''')$ for all $j$. Since the eigenvalues of a cycle graph of length $L (\Gamma)$ with continuity-Kirchhoff conditions coincide with those of the interval $(0, L (\Gamma))$ with periodic boundary conditions $f (0) = f (L (\Gamma))$ and $f' (0) = f' (L (\Gamma))$, we have
 %----------------%
\begin{align*}
 \lambda_{2 k} (\Gamma''') = \lambda_{2 k + 1} (\Gamma''') = \frac{4 k^2 \pi^2}{L (\Gamma)^2}, \quad k = 1, 2, \dots,
\end{align*}
 %----------------%
and the proof of the theorem is complete.
\end{proof}

We conclude with a remark on the sharpness of the estimates in Theorem \ref{thm:optimizationGeneral}.

\begin{remark}
As we have seen in Example \ref{ex:figure8}, for the equilateral figure-8 graph we have indeed $\lambda_1 (\Gamma) = -1$, $\lambda_2 (\Gamma) = 0$, as well as
\begin{align*}
 \lambda_{2 k + 1} (\Gamma) = \frac{4 k^2 \pi^2}{L (\Gamma)^2}, \quad k \in \N,
\end{align*}
which coincides with the estimate given in Theorem \ref{thm:optimizationGeneral} for the eigenvalues with odd index three or larger. Furthermore, for even $k$ we have $\lambda_{2 k} (\Gamma) = \lambda_{2 k + 1} (\Gamma) = \lambda_{2 k + 2} (\Gamma)$, due to the multiplicity of the eigenvalue. Hence, in general it is not possible to provide estimates for the eigenvalues $\lambda_{2 k} (\Gamma)$ better than estimating them trivially by $\lambda_{2 k + 1} (\Gamma)$.
\end{remark}

%%%%%%%%%%%%%%%%%%%%%%%%%%%%%%%%%%%%%%%%%
\subsection*{Data availability statement}

Data are available in the article.

%%%%%%%%%%%%%%%%%%%%%%%%%%%%%%%%%%%%%%%%%
\subsection*{Conflict of interest}

The authors have no conflict of interest.

%%%%%%%%%%%%%%%%%%%%%%%%%%%%%%%%%%%%%%%%%

\subsection*{Acknowledgments}

P.E. was supported by the European Union's Horizon 2020 research and innovation programme under the Marie Sk{\l}odowska-Curie grant agreement No 873071. J.R.\ was supported by the grant no.\ 2022-03342
of the Swedish Research Council (VR).


\begin{thebibliography}{99}

\bibitem{AEHK24p} M.\ Ahrami, Z.\ El Allali, E.\,M.\ Harrell II, and James B.\ Kennedy, {\it Optimizing the fundamental eigenvalue gap of quantum graphs}, preprint, arXiv:2401.04344.

% \bibitem{BEL22} M.~Baradaran, P.~Exner, and J.~Lipovsk\'y, {\it Magnetic ring chains with vertex coupling of a preferred orientation}, J.\ Phys.\ A 55 (2022), no. 37, Paper No. 375203.

\bibitem{BL17} R.~Band, G.~L\'evy, {\it Quantum graphs which optimize the spectral gap}, Ann. H.~Poincar\'{e} 18 (2017), 3269--3323.

\bibitem{BEL23} M.~Baradaran, P.~Exner, and J.~Lipovsk\'y, {\it Magnetic square lattice with vertex coupling of a preferred orientation}, Ann.\ Physics 454 (2023), Paper No.\ 169339.

\bibitem{BET22} M.~Baradaran, P.~Exner, and M.~Tater, {\it Spectrum of periodic chain graphs with time-reversal non-invariant vertex coupling}, Ann.\ Physics 443 (2022), Paper No. 168992.

\bibitem{BKKM17} G.\ Berkolaiko, J.\,B.\ Kennedy, P.\ Kurasov, and D.\ Mugnolo, {\it Edge connectivity and the spectral gap of combinatorial and quantum graphs}, J.\ Phys.\ A 50 (2017), no. 36, 365201, 29 pp.

\bibitem{BKKM19} G.\ Berkolaiko, J.\,B.\ Kennedy, P.\ Kurasov, and D.\ Mugnolo, {\it Surgery principles for the spectral analysis of quantum graphs}, Trans.\ Amer.\ Math.\ Soc.\ 372 (2019), no. 7, 5153--5197.

\bibitem{BK13} G.~Berkolaiko and P.~Kuchment, Introduction to Quantum Graphs, Mathematical Surveys and Monographs 186, American Mathematical Society, Providence, RI, 2013.

\bibitem{BCJ21} D.\ Borthwick, L.\ Corsi, and K.\ Jones, {\it Sharp diameter bound on the spectral gap for quantum graphs},
Proc.\ Amer.\ Math.\ Soc.\ 149 (2021), no. 7, 2879--2890.

\bibitem{BHY23} D.\ Borthwick, E.\,M.\ Harrell II, and H.\ Yu, {\it Gaps between consecutive eigenvalues for compact metric graphs}, J.\ Math.\ Anal.\ Appl.\ 531 (2024), no.\ 1, Paper No. 127802.

\bibitem{EL19} P.~Exner and J.~Lipovsk\'y, {\it Spectral asymptotics of the Laplacian on Platonic solids graphs}, J.\ Math.\ Phys.\ 60 (2019), no. 12, 122101.

\bibitem{ET18} P.~Exner and M.~Tater, {\it Quantum graphs with vertices of a preferred orientation}, Phys.\ Lett.\ A 382 (2018), no. 5, 283--287.

\bibitem{ET21} P.~Exner and M.~Tater, {\it Quantum graphs: self-adjoint, and yet exhibiting a nontrivial $\mathcal{PT}$-symmetry}, Phys.\ Lett.\ A 416 (2021),  Paper No.\ 127669

\bibitem{Fr05} L.~Friedlander, {\it Extremal properties of eigenvalues for a metric graph}, Annales de l'Institut Fourier 55 (2005), 199--211.

\bibitem{Ha00} M.~Harmer, {\it Hermitian symplectic geometry and extension theory}, J.\ Phys.\ A 33 (2000), no.\ 49, 9015--9032.

\bibitem{KKT16} G.~Karreskog, P.~Kurasov, and I.~Trygg Kupersmidt, {\it Schr\"odinger operators on graphs: symmetrization and Eulerian cycles}, Proc. AMS 144 (2016), 1197--1207.

\bibitem{K95} T.~Kato, Perturbation Theory for Linear Operators, Springer-Verlag, Berlin, 1995.

\bibitem{K20} J.\,B.\ Kennedy, {\it A family of diameter-based eigenvalue bounds for quantum graphs}, Discrete and continuous models in the theory of networks, 213--239, Oper.\ Theory Adv.\ Appl.\ 281, Birkhäuser/Springer, Cham, 2020.

\bibitem{KKMM16} J.B.~Kennedy, P.~Kurasov, P., G.~Malenov\'{a}, and D.~Mugnolo, {\it On the spectral gap of a quantum graph}, Ann. H.~Poincar\'{e} 17, (2016), 1--35.

\bibitem{KN22} A.~Kostenko and N.~Nicolussi, Laplacians on Infinite Graphs, Mem.~EMS, vol.~3, EMS Press, Berlin 2022.

\bibitem{KS99} V.~Kostrykin and R.~Schrader, {\it Kirchhoff's rule for quantum wires}, J. Phys. A: Math. Gen. 32 (1999), 595--630.

\bibitem{Ku24} P.~Kurasov, Spectral Geometry of Graphs, Oper.\ Theory Adv.\ Appl.~293, Birkh\"auser/Springer, Berlin, 2024.

\bibitem{KMN13} P.~Kurasov,  G.~Malenov\'{a}, and S.~Naboko, {\it Spectral gap for quantum graphs and their edge connectivity}, J. Phys. A 46 (2013), Paper No. 275309.

\bibitem{KN14} P.~Kurasov and S.~Naboko, {\it Rayleigh estimates for differential operators on graphs}, J. Spect. Theory 4 (2014), 211--219.

\bibitem{KS19} P.\ Kurasov and A.\ Serio, {\it Optimal potentials for quantum graphs}, Ann.\ Henri Poincar\'e 20 (2019), no.\ 5, 1517--1542.

\bibitem{M14} D.\ Mugnolo, Semigroup Methods for Evolution Equations on Networks, Underst.\ Complex Syst., Springer, Cham, 2014.

\bibitem{MR21} J.\ Muller and J.\ Rohleder, {\it The Krein--von Neumann extension for Schr\"odinger operators on metric graphs}, Complex Anal.\ Oper.\ Theory 15 (2021), no.2, Paper No. 27.

\bibitem{Ni87} S.~Nicaise, {\it Spectre des r\'eseaux topologiques finis}, Bull. Sci. Math. (2) 111 (1987), 401--413.

\bibitem{P3+05} Y.V.~Pokornyi, O.\,M.~Penkin, V.\,L.~Pryadiev, A.\,V.~Borovskikh, P.\,P.~Lazarev, and S.\,A.~Shabrov, Differential Equations on Geometric Graphs, Fizmatlit, Moscow, 2005.

\bibitem{RS78} M.~Reed and B.~Simon, Methods of Modern Mathematical Physics, IV.~Analysis of Operators, Academic Press, New York 1972-1979.

\bibitem{RB69} F.\,S.\ Rofe-Beketov, {\it Self-adjoint extensions of differential operators in a space of vector-valued functions}, Teor.\ Funkcii, Funkcional.\ Anal.\ Prilozh.\ 8  (1969), 3--24 (in Russian).

\bibitem{Ro17} J.~Rohleder, {\it Eigenvalue estimates for the Laplacian on a metric tree}, Proc.\ Amer.\ Math.\ Soc.\ 145 (2017), 2119--2129.

\bibitem{R22} J.~Rohleder, {\it Quantum trees which maximize higher eigenvalues are unbalanced}, Proc.\ Amer.\ Math.\ Soc.\ Ser.\ B 9 (2022), 50--59.

\bibitem{RS20} J.\ Rohleder and C.\ Seifert, {\it Spectral monotonicity for Schr\"odinger operators on metric graphs}, Discrete and continuous models in the theory of networks, 291--310, Oper.\ Theory Adv.\ Appl.\ 281, Birkhäuser/Springer, Cham, 2020.

\bibitem{SK15} P.~St\v{r}eda and J.~Ku\v{c}era, {\it Orbital momentum and topological phase transformation}, Phys.\ Rev.\ B92 (2015), Paper No. 235152.

\end{thebibliography}
\end{document}